\documentclass[12pt]{iopart}
\usepackage{amsfonts}
\usepackage{graphicx}  
\usepackage{iopams}  
\usepackage{hyperref}
\usepackage{enumerate}

\newtheorem{theorem}{Theorem}[section]
\newtheorem{lemma}[theorem]{Lemma}

\newtheorem{corollary}[theorem]{Corollary}

\newtheorem{remark}{Remark}[section]

\newenvironment{proof}[1][Proof]{\begin{trivlist}
\item[\hskip \labelsep {\bfseries #1}]}{\end{trivlist}}

\newcommand{\qed}{\nobreak \ifvmode \relax \else
      \ifdim\lastskip<1.5em \hskip-\lastskip
      \hskip1.5em plus0em minus0.5em \fi \nobreak
      \vrule height0.75em width0.5em depth0.25em\fi}

\begin{document}

\title[Cosmic Censorship for Self-Similar Spherical Dust Collapse]{Cosmic Censorship for Self-Similar Spherical Dust Collapse}

\author{Emily M. Duffy and Brien C. Nolan}

\address{School of Mathematical Sciences, Dublin City University, Glasnevin, Dublin 9, Ireland.}
\eads{\mailto{emilymargaret.duffy27@mail.dcu.ie}, \mailto{brien.nolan@dcu.ie}}
\begin{abstract}
We undertake a rigorous study of the stability of the Cauchy horizon in the naked self-similar Lema\^{i}tre-Tolman-Bondi spacetimes under even parity linear perturbations. We use a combination of energy methods and results about $L^p$-spaces to determine the behaviour of the perturbations as they evolve through the spacetime. We first establish that an average of the perturbation generically diverges on the Cauchy horizon. We next introduce a rescaled version of the perturbation, and show that it is bounded and non-zero on the Cauchy horizon. This in turn shows that the perturbation itself diverges in a pointwise fashion on the Cauchy horizon. We give a physical interpretation of this result using the perturbed Weyl scalars. This result supports the hypothesis of cosmic censorship. 

\end{abstract}

\maketitle

\section{Introduction: Cosmic Censorship and Perturbation Theory}
\label{sec:intro}

The formation of naked singularities in various collapse models is a well known occurance; examples include the Reissner-Nordstr\"{o}m (RN) and Kerr spacetimes \cite{Wald}, certain models of critical collapse \cite{LivRev} and certain self-similar perfect fluid and dust solutions \cite{CCreview}. In response to the formation of such singularities, Roger Penrose proposed that the evolution of physically reasonable and generic initial data will not result in the formation of a naked singularity, a statement known as the cosmic censorship hypothesis (CCH) \cite{Wald}. Indeed, it is noted that many spacetimes which contain naked singularities also show high degrees of symmetry and it is tempting to conclude that the formation of the naked singularity is (in some cases) due to the unphysical symmetry of the spacetime. 

It follows that one manner in which the CCH can be studied is by disturbing the symmetry of a naked singularity spacetime by introducing perturbations. Should these perturbations remain finite as they evolve through the spacetime, then the singularity has displayed stability to such perturbations. In particular, the behaviour of perturbations on the Cauchy horizon (CH) of the RN spacetime is illustrative. Recall that the CH is the first null ray emitted by the singularity; it can equally be thought of as the boundary which separates observers who can see the singularity from those who cannot. In RN, the CH is a hypersurface on which perturbations diverge. More precisely, metric perturbations which arrive at the CH from the exterior have an infinite flux on the CH, as measured by observers crossing the CH. One might expect that perturbations evolving through other spacetimes containing naked singularities would show similar divergent behaviour on the CH. 

Should the perturbations turn out to behave in a finite manner on the CH, in certain cases one can still rule out the spacetime as a serious counter-example to the CCH.  In some cases, the spacetime involves an unrealistic matter model, for example the Vaidya spacetime containing null dust or the perfect fluid spacetime, which neglects shear, viscosity and heat conduction. Another phenomenon which sometimes occurs is that the formation of the naked singularity depends on specific initial data; any perturbation of the initial data and the naked singularity fails to form \cite{Christod}. 

We consider here the behaviour of even parity perturbations of the self-similar Lema\^{i}tre-Tolman-Bondi spacetime. We note that this spacetime cannot be taken as a serious counter-example to the CCH, as the first of the above-mentioned defects is present in this spacetime. The matter model used is a dust cloud which ignores pressure (and pressure gradients), and therefore does not provide a realistic description of gravitational collapse. Nonetheless, the simplicity of this spacetime makes it a very useful toy model for gravitational collapse resulting in naked singularity formation. For example, in the present paper we develop techniques that should in principle be applicable to any spherically symmetric self-similar spacetime. 

In a previous paper \cite{DN}, the behaviour of odd parity perturbations of the self-similar LTB spacetime was considered. These perturbations were found to remain finite at the CH, where finiteness is measured with respect to certain energy norms of the perturbation, and pointwise values thereof. This result holds for a general choice of initial data and initial data surface. In this paper, we use broadly similar techniques, with non-trivial extensions, to examine the behaviour of even parity perturbations of the same spacetime. 

The question of the behaviour of even parity perturbations of this spacetime was considered by Nolan and Waters \cite{NW}. In this work, both a harmonic decomposition and a Fourier mode decomposition of the gauge invariant perturbation variables and equations were used, and the individual Fourier modes $X_{\omega,\ell,m}(z)$ were analysed. It was found that modes which are finite on the past null cone remain finite on the Cauchy horizon. However, the question of how to resum these modes on the Cauchy horizon was not fully addressed and so the problem of getting a complete analytic understanding the even parity linear perturbations remains unresolved.

In the next section, we describe the structure of the self-similar LTB spacetime and examine conditions necessary for this spacetime to have a naked singularity. In Section \ref{sec:GS}, we discuss the perturbation formalism due to Gerlach and Sengupta. This formalism uses a multipole decomposition to write linear metric and matter perturbations in terms of functions depending on the similarity coordinate $z$ and the comoving radial coordinate $r$. These perturbations do not share the background symmetries of self-similarity and spherical symmetry. In Section \ref{sec:reduction} we derive a fundamental four-dimensional system of coupled linear PDEs which describe the evolution of the even parity perturbations. 

In practice, we work with a five dimensional system, which has the very useful property of symmetric hyperbolicity. These five variables obey an equation of motion arising from the linearised Einstein equations and must also satisify a constraint, the elimination of which produces the four variable system. Having determined a useful form for the perturbation variables and the Einstein equations, we aim to examine the behaviour of the perturbations as they reach the CH. In particular, we wish to know whether or not the perturbations remain finite there. 

Our strategy in addressing this problem is as follows. In Section \ref{sec:reduction}, we find that the perturbations obey a PDE whose coefficients, as a consequence of the self-similarity, depend only on the similarity coordinate. This means that we can take the state vector $\vec{u}$ which describes the perturbations and integrate with respect to the radius, to produce a kind of ``average''. Since the CH is a surface of constant similarity variable, and since the equations of motion have coefficients independent of the radius, we can reasonably expect that the behaviour of this averaged perturbation should reflect the behaviour of the perturbation itself. In other words, we use the behaviour of solutions to the averaged perturbation's equation of motion as a guide to the behaviour of the perturbation itself. This is the core strategy behind this work. 

In Section \ref{sec:lqblowup} we analyse the behaviour of solutions to the ODE which this averaged form satisfies, as they approach the CH. This ODE has a regular singular point, and methods for solving such systems are well understood. We can show that solutions to this system generically blow-up on the CH, with a characteristic power given by an eigenvalue of a particular matrix, which we denote $-c$. However, since this result applies only to an average of $\vec{u}$, we cannot immediately conclude that $\vec{u}$ blows-up at the CH in a pointwise manner. 

In Section \ref{sec:chdivbehaviour}, we investigate the pointwise behaviour of $\vec{u}$ by introducing a new state vector $\vec{x}$, related to $\vec{u}$, which we expect to have a finite limit on the CH. We devote the next section to a series of results, which cumulatively show that $\vec{x}$ is finite and non-zero on the CH. This in turn establishes that $\vec{u}$ blows-up in a pointwise manner on the CH. 

Finally, in Section \ref{sec:gi} we use the perturbed Weyl scalars to give a physical interpretation of our result. In Section \ref{sec:concl} we make some concluding remarks and discuss further developments of this work. We use units in which $G=c=1$. 


\section{The Self-Similar LTB Spacetime }
\label{sec:background}
\subsection{The LTB Spacetime}
\label{sec:ltb}
The Lema\^{i}tre-Tolman-Bondi spacetime is a spherically symmetric spacetime containing a pressure-free perfect fluid which undergoes an inhomogeneous collapse into a singularity. Under certain conditions this singularity can be naked. We will initially use comoving coordinates $(t,r, \theta, \phi)$, in which the radius $r$ labels each successive shell in the collapsing dust. The line element for such a spacetime can be written as 
\begin{equation}
ds^2=-dt^2 + \rme^{\nu}(t,r)dr^2 + R^2(t,r)d\Omega^2,
\label{metric_tr}
\end{equation}
where $d\Omega^2=d\theta^2+\sin^2\theta d\phi^2$ and $R(t,r)$ is the physical radius of the dust. The stress-energy tensor of the dust can be written as 
\begin{equation*}
\bar{T}^{\mu \nu} = \bar{\rho}(t,r) \bar{u}^{\mu} \bar{u}^{\nu},
\end{equation*}
where $\bar{u}^{\mu}$ is the 4-velocity of the dust and $\bar{\rho}(t,r)$ is its rest mass density. In comoving coordinates, $\bar{u}^{\mu} = \delta_{0}^{\mu}$. 

The background Einstein equations for the metric and stress energy in comoving coordinates immediately provide the following results,
\begin{equation}
\label{rhoande}
\rme^{\nu /2}=\frac{1}{\sqrt{1+f(r)}} \frac{\partial R}{\partial r} , \qquad   \bar{\rho}(t,r)=\frac{1}{4 \pi R^2} \left(\frac{\partial R}{\partial r}\right)^{-1} \frac{\partial m}{\partial r},  
\end{equation}
\begin{equation}
\label{Rdot}
\left( \frac{\partial R}{\partial t} \right)^2-\frac{2 m(r)}{R}=f(r). 
\end{equation}  
The function $m(r)$ is known as the Misner-Sharp mass and is a suitable mass measure for spherically symmetric spacetimes.  

Recall that a shell focusing singularity is a singularity which occurs when the physical radius $R(t,r)$ of the dust cloud decreases to zero, so that all the matter shells have been ``focused'' onto a single point. In this spacetime, a shell focusing singularity occurs on a surface of the form $t=t_{sf}(r)$, which includes the scaling origin $(t,r)=(0,0)$. 

We immediately specialise to the marginally bound case by setting $f(r)=0$. See \cite{DN} for more details on the background spacetime. 


\subsection{Self-Similarity}
\label{sec:selfsim}
We follow here the conventions of \cite{CC}. A spacetime displays self-similarity if it admits a homothetic Killing vector field, that is, a vector field $\vec{\xi}$ such that 
\begin{equation} \label{selfsim}
\mathcal{L}_{\vec{\xi}} g_{\mu \nu} = 2 g_{\mu \nu}. 
\end{equation}
In comoving coordinates, the homothetic Killing vector field is given by $\vec{\xi}=t\frac{\partial}{\partial t}+ r \frac{\partial}{\partial r }$. When self-similarity is imposed on the metric and stress-energy tensor, we find that functions appearing in the metric, the dust density and the Misner-Sharp mass have the following scaling behaviour
\begin{eqnarray} \label{nuandR}
\nu(t,r)=\nu(z), \qquad \qquad \qquad
R(t,r)=rS(z),
\end{eqnarray}
\begin{eqnarray} \label{rhoandm}
\bar{\rho}(t,r)=\frac{q(z)}{r^2}, \qquad \qquad \qquad
m(r)=\lambda r,
\end{eqnarray}
where $z=-t/r$ is the similarity variable and $\lambda$ is a constant (the case $\lambda=0$ corresponds to flat spacetime). Having imposed self-similarity, we find that combining (\ref{Rdot}), (\ref{nuandR}) and (\ref{rhoandm}), produces
\begin{equation} \label{Sdef}
S(z)= (a z +1)^{2/3},
\end{equation}
where $a=3\sqrt{\frac{\lambda}{2}}$. The form of the line element (\ref{metric_tr}) is invariant under the coordinate transformation $r \rightarrow \tilde{r} = \tilde{r}(r)$. We choose $r$ so that $R|_{t=0}=r$ and $\frac{\partial R}{\partial r}|_{t=0}=1$. With this expression for $S(z)$ we can solve for $\frac{\partial R}{\partial r}$ explicitly. In (\ref{rhoande}) we convert $\frac{\partial R}{\partial r}$ to a derivative in $(z,r)$ and find that 
\begin{equation} \label{expform}
\rme^{\nu/2}=\frac{\partial R}{\partial r}=(\frac{1}{3}az+1)(1+az)^{-1/3}. 
\end{equation}
We note that by combining (\ref{rhoandm}) with (\ref{rhoande}), we can find an expression for $q(z)$, 
\begin{equation}
\label{qform}
q(z) = \frac{a^2}{6 \pi (3+4az+a^2z^2)}. 
\end{equation}
We state the metric in $(z,r)$ coordinates for future use,

\begin{equation} \label{metriczr}
ds^2=-r^2 dz^2+\rme^{\nu}(z)(1-z^2\rme^{-\nu}(z))dr^2-2rz drdz + R^2d\Omega^2.  
\end{equation} 
In Section \ref{sec:gi} we will need the radial null directions of the self-similar LTB spacetime. In terms of $(z,r)$ coordinates, retarded null coordinates $u$ and $v$ take the form
\begin{eqnarray} \label{uvcoords}
u=r \, \exp \left( - \int_{z}^{z_{o}} \frac{dz'}{f_{+}(z')} \right), \qquad \quad
v=r \, \exp \left( - \int_{z}^{z_{o}} \frac{dz'}{f_{-}(z')}\right),
\end{eqnarray}
where $f_{\pm}:=\pm \rme^{\nu/2} + z$. In these coordinates, the metric takes the form

\begin{equation*} 
ds^2=- \frac{t^2}{uv} (1-\rme^{\nu}z^{-2}) \, du \, dv + R^2(t,r) d \, \Omega^2. 
\end{equation*}
In  order to calculate the perturbed Weyl scalars, we will need the radially in- and outgoing null vectors, $l^{\mu}$ and $n^{\mu}$. These vectors obey the normalization $g_{\mu \nu} l^{\mu} n^{\nu}=-1$. A suitable choice is therefore
\begin{equation}
\label{inoutvectors}
\vec{l}=\frac{1}{B(u,v)} \frac{\partial}{\partial u}, \qquad \qquad \vec{n}=\frac{\partial}{\partial v}, 
\end{equation}
where $B(u,v)=\frac{t^2}{2uv} \left( 1-\frac{\rme^{\nu(z)}}{z^2} \right)$. In what follows, we shall always use a dot to indicate differentiation with respect to the similarity variable $z$, $\cdot = \frac{\partial}{\partial z}$. 
 

\subsection{Nakedness of the Singular Origin}
\label{nakedsing}
We now consider the conditions required for the singularity at the scaling origin $(t,r)=(0,0)$ to be naked. As a necessary and sufficient condition for nakedness, the spacetime must admit causal curves which have their past endpoint on the singularity. It can be shown \cite{NW} that it is actually sufficient to consider only null geodesics with their past endpoints on the singularity, and without loss of generality, we restrict our attention to the case of radial null geodesics (RNGs). 

The equation which governs RNGs can be read off the metric (\ref{metric_tr}), 
\begin{equation} \label{RNGS}
\frac{dt}{dr}=\pm \rme^{\nu /2}. 
\end{equation}
Since we wish to consider outgoing RNGs we select the $+$ sign. We can convert the above equation into an ODE in the similarity variable, 
\begin{equation}\label{zrrngs}
z+rz'=-\rme^{\nu /2}. 
\end{equation}
We look for constant solutions to this equation, which correspond to null geodesics that originate from the singularity. In other words, the existence of constant solutions to (\ref{zrrngs}) is equivalent to the nakedness of the singularity. For constant solutions, we set the derivative of $z$ to zero, so that the Cauchy horizon is the first root of 
\begin{equation} \label{CH}
z+\rme^{\nu /2}= z+S(z)-z \dot{S}(z)=0, 
\end{equation}
where we reserve the dot notation to indicate a derivative with respect to $z$, $\cdot:=d / dz$. By using (\ref{Sdef}), we can write (\ref{CH}) as an algebraic equation in $z$, 
\begin{equation} \label{algz}
az^4+\left(1+\frac{a^3}{27}\right)z^3+\left(\frac{a^2}{3}\right)z^2 +az+1=0. 
\end{equation}
We wish to discover when this equation will have real solutions. This can easily be found using the polynomial discriminant for a quartic equation, which is negative when there are two real roots. In this case we have
\begin{equation}
D=\frac{1}{27}(-729+2808a^3-4a^6),
\end{equation}
which is negative in the region $a<a^*$ where $a^*$ is
\begin{equation}
a^*=\frac{3}{(2(26+15\sqrt{3}))^{1/3}} \approx 0.638 \ldots
\end{equation}
This translates to the bound $\lambda \leq 0.09$. From (\ref{rhoandm}), we can see that this result implies that singularities which are ``not too massive'' can be naked. See Figure \ref{Fig1} for a Penrose diagram of this spacetime. 

\begin{remark} {\em In fact, one can find $D < 0$ in two ranges, namely $a<a^* \approx 0.64$ and $a > a^{**} \approx 8.89$. We reject the latter range as begin unphysical. Consider (\ref{Sdef}), which indicates that the shell-focusing singularity occurs at $z=-1/a$. If we chose the range $a>a^{**}$ we would find that the corresponding outgoing RNG occurs after the shell focusing singularity and so is not part of the spacetime.}
\end{remark}

\begin{remark} {\em We note that this analysis has assumed that the entire spacetime is filled with a dust fluid. A more realistic model would involve introducing a cutoff at some radius  $r=r*$, after which the spacetime would be empty. We would then match the interior matter-filled region to an exterior Schwarszchild spacetime. However, it can be shown that this cutoff spacetime will be globally naked so long as the cutoff radius is chosen to be sufficiently small \cite{Joshi}. We will therefore neglect to introduce such a cutoff. }
\end{remark}

\begin{figure} 
\begin{center}
\includegraphics[scale=1]{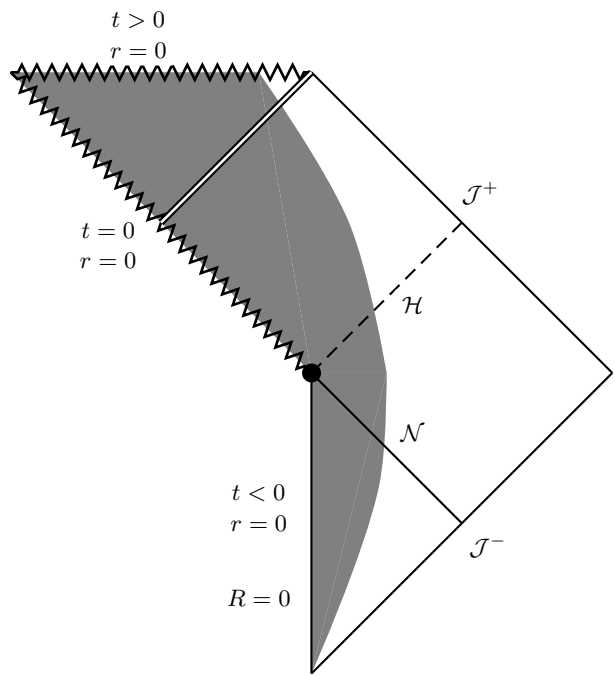}
\end{center}
\caption{Structure of the self-similar LTB spacetime. We present here the conformal diagram for the self-similar LTB spacetime. The gray shaded region represents the interior of the collapsing dust cloud. We label the past null cone of the naked singularity by $\mathcal{N}$, future and past null infinity by $\mathcal{J}^{+}$ and $\mathcal{J}^{-}$ and the Cauchy horizon by $\mathcal{H}$. 
}
\label{Fig1}
\end{figure}

\section{The Gerlach-Sengupta Formalism} 
\label{sec:GS}
We shall use the Gerlach-Sengupta method \cite{GS} to perturb this spacetime (we follow the presentation of \cite{MGGundlach}). This method exploits the spherical symmetry of the spacetime by performing a decomposition of the spacetime into two submanifolds (with corresponding metrics).  Perturbations of the spacetime are then expanded in a multipole decomposition and gauge invariant combinations of the perturbations are constructed. 

We begin by writing the metric of the entire spacetime $(\mathcal{M}^4, g_{\mu\nu})$ as
\begin{equation} \label{gsmetric}
ds^2=g_{AB}(x^{C}) dx^A dx^B + R^2(x^{C}) \gamma_{ab} dx^adx^b,
\end{equation}
where $g_{AB}$ is a Lorentzian metric on the 2-dimensional manifold $\mathcal{M}^2$ and $\gamma_{ab}$ is the metric for the 2-sphere $\mathcal{S}^2$ 
(and the full manifold is $\mathcal{M}^4=\mathcal{M}^2 \times \mathcal{S}^2$). The indices $A, B, C \ldots$ indicate coordinates on $\mathcal{M}^2$ and take the values $A, B \ldots =1,2$ while the indices $a, b, c \ldots$ indicate coordinates on $\mathcal{S}^2$ and take the values $a,b \ldots=3,4$. The covariant derivatives on $\mathcal{M}^4$, $\mathcal{M}^2$ and $\mathcal{S}^2$ are denoted by a semi-colon, a vertical bar and a colon respectively. The stress-energy can be split in a similar fashion,
\begin{equation*}
t_{\mu \nu}dx^{\mu} dx^{\nu}=t_{AB}dx^Adx^B + Q(x^C)R^2\gamma_{ab}dx^adx^b,
\end{equation*}
where $Q(x^C)=\frac{1}{2}t^a_{\, \, a}$ is the trace across the stress-energy on $\mathcal{S}^2$, which vanishes in the LTB case. Now if we define 
\begin{equation*}
v_{A}=\frac{R_{|A}}{R},
\end{equation*}
\begin{equation*}
V_{0}=-\frac{1}{R^2} + 2v^{A}_{\, \, |A} + 3 v^{A}v_{A},
\end{equation*}
then the Einstein equations for the background metric and stress-energy read 
\begin{equation*}
G_{AB}=-2(v_{A|B}+ v_{A}v_{B})+V_{0}g_{AB}=8 \pi t_{AB},
\end{equation*}
\begin{equation*}
\frac{1}{2}G^{a}_{\, \, a}=-\mathcal{R}+v^{A}v_{A}+v^{A}_{\, \, \,|A}=8 \pi Q(x^{C} ),
\end{equation*}
where $G^{a}_{\, \, a}=\gamma^{ab} G_{ab}$ and $\mathcal{R}$ is the Gaussian curvature of $\mathcal{M}^2$, $\mathcal{R}=\frac{1}{2} R_{A}^{(2)A}$ where $R^{(2)}_{AB}$ indicates the Ricci tensor on $\mathcal{M}^2$.  

We now wish to perturb the metric (\ref{gsmetric}), so that $g_{\mu \nu}(x^{\alpha}) \rightarrow g_{\mu \nu}(x^{\alpha})+ \delta g_{\mu \nu}(x^{\alpha})$. To do this, we decompose $\delta g_{\mu \nu}(x^{\alpha})$ and write explicitly the angular dependence using the spherical harmonics. We write the spherical harmonics as $Y^{m}_{l}\equiv Y$. $\{Y\}$  forms a basis for scalar harmonics, while $\{Y_{a}:=Y_{:a}, S_{a}:=\epsilon^{b}_{a}Y_{b}  \}$ form a basis for vector harmonics. Finally, $\{Y\gamma_{ab}, Z_{ab}:=Y_{a:b}+\frac{l(l+1)}{2}Y \gamma_{ab}, S_{a:b}+S_{b:a} \}$ form a basis for tensor harmonics.

We can classify these harmonics according to their behaviour under spatial inversion $\vec{x} \rightarrow -\vec{x}$: A harmonic with index $l$ is even if it transforms as $(-1)^{l}$ and odd if it transforms as $(-1)^{l+1}$. According to this classification, $Y$, $Y_{a}$ and $Z_{ab}$ are even, while $S_{a}$ and $S_{(a:b)}$ are odd. 

We now expand the metric perturbation in terms of the spherical harmonics. Each perturbation is labelled by $(l,m)$ and the full perturbation is given by a sum over all $l$ and $m$. However, since each individual perturbation decouples in what follows, we can neglect the labels and summation symbols. The metric perturbation is given by 
\begin{equation*}
\delta g_{AB}=h_{AB} Y, \qquad \qquad \delta g_{Ab}=h^{\scriptsize{\textbf{E}}}_{A}Y_{:b}+h^{\scriptsize{\textbf{O}}}_{A}S_{b},
\end{equation*}
\begin{equation*}
\delta g_{ab}=R^2KY \gamma_{ab}+R^2GZ_{ab}+h(S_{a:b}+S_{b:a}),
\end{equation*}
where $h_{AB}$ is a symmetric rank 2 tensor,  $h^{\scriptsize{\textbf{E}}}_{A}$ and $h^{\scriptsize{\textbf{O}}}_{A}$ are vectors and $K$, $G$ and $h$ are scalars, all on $\mathcal{M}^2$. We similarly perturb the stress-energy $t_{\mu \nu} \rightarrow t_{\mu \nu}+ \delta t_{\mu \nu}$ and expand the perturbation in terms of the spherical harmonics,
\begin{equation} \label{m2se}
\delta t_{AB}=\Delta t_{AB} Y, \qquad \qquad \delta t_{Ab}=\Delta t^{\scriptsize{\textbf{E}}}_{A} Y_{:b} + \Delta t^{\scriptsize{\textbf{O}}}_{A} S_{b},
\end{equation}
\begin{equation} \label{s2se}
\delta t_{ab}=r^2 \Delta t^3 \gamma_{ab} Y + r^2 \Delta t^2 Z_{ab} + 2\Delta t S_{(a:b)},
\end{equation}
where $\Delta t_{AB}$ is a symmetric rank 2 tensor,  $\Delta t^{\scriptsize{\textbf{E}}}_{A}$ and $\Delta t^{\scriptsize{\textbf{O}}}_{A}$ are vectors and $\Delta t^3$, $\Delta t^2$ and $\Delta t$ are scalars, all on $\mathcal{M}^2$. 

We wish to work with gauge invariant variables, which can be constructed as follows. Suppose the vector field $\vec{\xi}$ generates an infinitesimal coordinate transformation $\vec{x} \rightarrow \vec{x'}=\vec{x}+ \vec{\xi}$. We wish our variables to be invariant under such a transformation. We can decompose $\vec{\xi}$ into even and odd harmonics and write the one-form fields 
\begin{equation*}
\underline{\xi}^E=\xi_{A}(x^C)Ydx^A + \xi^{\scriptsize{\textbf{E}}}(x^C)Y_{:a} dx^a, \qquad \qquad \underline{\xi}^{\scriptsize{\textbf{O}}}=\xi^{\scriptsize{\textbf{O}}} S_{a} dx ^a. 
\end{equation*}
We then construct the transformed perturbations after this coordinate transformation and look for combinations of perturbations which are independent of $\vec{\xi}$, and therefore gauge invariant. We will list here only the even parity gauge independent perturbations and we set the odd parity perturbations to zero. In the even parity case, the metric perturbation is described by a gauge invariant 2-tensor $k_{AB}$ and a gauge invariant scalar $k$,
\begin{equation} 
\label{ktensor}
k_{AB}=h_{AB}-(p_{A|B}+p_{B|A}), \qquad \qquad k=K-2v^{A}p_{A},
\end{equation}
where $p_{A}=h_{A}^{\scriptsize{\textbf{E}}}-(1/2)R^2G_{|A}$. The even parity gauge invariant matter perturbation is given by
\begin{equation} \label{tabdef}
T_{AB}=\Delta t_{AB}-{t_{AB}}^{|C}p_{C} -2(t_{CA}{p^{C}}_{|B} + t_{CB}{p^C}_{|A}),
\end{equation} 
\begin{equation} \label{TAdef}
T_{A}=\Delta t_{A} -{t_{A}}^{C}p_{C} -R^2 \left( \frac{{t^a}_{a}}{4} \right) G_{|A},
\end{equation}
\begin{equation} \label{tscaldef}
T^3=\Delta t^3 - \frac{p^C}{R^2} \left( \frac{R^2 {t_{a}}^a}{2} \right)_{|C},   \qquad \qquad T^2=\Delta t^2 -  \left( \frac{R^2 {t_{a}}^a}{2} \right) G. 
\end{equation}
In the Regge-Wheeler gauge, $h^{E}_{A}=h=G=0$, which implies that $p_{A}=0$ and the gauge invariant matter perturbations coincide with the bare matter perturbations, which simplifies matters considerably. 

We now list the linearised Einstein equations, which can be written in terms of the gauge invariant quantities listed above. We find
\begin{eqnarray} 
(k_{CA|B}+k_{CB|A}-k_{AB|C})v^C-g_{AB}(2k_{CD}^{ \quad|D}-k_{D \, |C}^{\, \, D})v^{C}   \nonumber \\
\qquad-(k_{|A}v_{B}+k_{|B}v_{A}+k_{|AB}) + \left( V_{0} + \frac{l(l+1)}{2R^2} \right) k_{AB}  \nonumber \\
\qquad\qquad-g_{AB} \left( k_{F}^{\, \, F} \frac{l(l+1)}{2R^2} + 2 k_{DF}v^{D|F}   +3k_{DF}v^Dv^F  \right)  \nonumber \\
\label{linein1}
\qquad \qquad \qquad + g_{AB} \left( \frac{(l-1)(l+2)}{2R^2} k - k_{\, |F}^{F}-2k_{|F}v^F\right)&= 8 \pi T_{AB},    
\end{eqnarray}
\begin{eqnarray}
-k^{AB}_{\quad |AB} + (k_{A}^{\, \, A})_{|B}^{\, \, \, |B} -2{k^{AB}}_{|A} v_{B} + k^{A}_{\, \, A|B}v^{B} + R^{AB}(k_{AB}-kg_{AB}) \nonumber \\
 \label{linein2}
\qquad \qquad \qquad \qquad \qquad  -\frac{l(l+1)}{2R^2} k^A_{\, \, A} +  k_{|A}^{\, \, \, \, A} + 2k_{|A} v^A = 16 \pi T^3 &,    \\
\label{linein3}
k_{AB}^{\quad |B} -k_{\,\, B|A}^{B} + k^B_{\, \, B} v_{B} -k_{|A} = 16 \pi T_{A}, &    \\
\label{linein4}
k^A_{\, \, A}=16 \pi T^2. &  
\end{eqnarray}
Finally, we present the linearised stress-energy conservation equations, for completeness. 
\begin{eqnarray} 
\label{secons1}
\fl (t_{AB|D} + 2t_{AB}v_{D})k^{BD} +  Q(k_{|A}-2kv_{A}) - t_{AB}k^{|B} + \frac{1}{2}t^{BC}k_{BC|A} - \frac{1}{2}t_{AB}k_{F}^{\, \, F|B}  \nonumber \\
\fl \qquad \qquad \qquad \qquad +t_{AB}k^{BF}_{\quad |F},  -\frac{1}{R^2}(R^2 T_{AB})^{|B} + \frac{l(l+1)}{R^2} T_{A} + 2 v_{A} T^3 =0  \\ 
\label{secons2}
\fl \frac{1}{R^2} (R^2T_{B})^{|B}+T^3 -\frac{(l-1)(l+2)}{2} \frac{T^2}{R^2} = \frac{1}{2} k_{AB}t^{AB}+ Q(k-\frac{1}{2}k_{A}^{\, \, A}). 
\end{eqnarray}
We now consider the form of these equations for the background LTB spacetime. 
\subsection{The Matter Perturbation}
\label{sec:matter}
To proceed further, we must find a relation between the gauge invariant matter perturbation terms (\ref{tabdef} - \ref{tscaldef}) and the dust density and velocity discussed in section 2.1. This amounts to specifying the matter content of the perturbed spacetime. To do this, we write the stress-energy of the full spacetime as a sum of the background stress-energy and the perturbation stress-energy (where a bar indicates a background quantity),
\begin{equation*}
T_{\mu \nu}=\overline{T}_{\mu \nu}+\delta T_{\mu \nu}. 
\end{equation*}
We assume that the perturbed spacetime also contains dust and write the density as $\rho=\overline{\rho}+\delta \rho$ and the fluid velocity as $u_{\mu}=\overline{u}_{\mu}+ \delta u_{\mu}$. We can now find an expression for the perturbation stress-energy (keeping only first order terms),
\begin{equation} \label{dustsepert}
\delta T_{\mu \nu}=\overline{\rho}(\overline{u}_{\mu} \delta u_{\nu}+ \overline{u}_{\nu} \delta u_{\mu})+ \delta \rho \overline{u}_{\mu} \overline{u}_{\nu}. 
\end{equation}
The perturbation of the dust velocity can now be expanded in terms of spherical harmonics as $\delta u_{\mu}=(\delta u_A, \delta u_{a})=( \delta u_A Y,  \delta u_{E} Y_{:a})$. By imposing conservation of stress-energy and requiring that the perturbed velocity $ u_{\mu}=\overline{u}_{\mu} + \delta u_{\mu}$ obeys $u^{\mu} u_{\mu}=-1$, one can show that the perturbed velocity can be written in the form
\begin{equation} \label{delu}
\delta u_{\mu} = (\partial_{A} \Gamma(t,r) Y, \gamma(t,r) Y_{:a}),
\end{equation}
where the variable $\Gamma$ acts as a velocity potential and obeys an equation of motion arising from perturbed stress-energy conservation (specifically, from the acceleration equations arising from stress-energy conservation), 
\begin{equation}  \label{Gaeqn}
\frac{\partial \Gamma}{\partial z} = -\frac{1}{2} \alpha(z,r),
\end{equation}
where we have labelled the first component of $k_{AB}$ (see (\ref{ktensor})) as $k_{00}=\alpha(z,p)$. In addition, by using (\ref{secons2}) one can show that 
\begin{equation} \label{littlega}
\gamma(z,p)=\Gamma(z,p) + g(p),
\end{equation}
where $p=\ln r$ and $g(p)$ is an initial data function for the velocity perturbation. If we compare this form for the perturbed stress-energy to the Gerlach-Sengupta form, we can find the gauge invariant matter perturbations for the LTB spacetime, which we write in terms of $(z,p)$ coordinates. The gauge invariant tensor $T_{AB}$ is given in $(z,p)$ coordinates by 
\begin{eqnarray*} 
T_{00}=2 \rho \rme^{p} \frac{\partial \Gamma}{\partial z} + \rme^{2p} \delta \rho,  \\
T_{01}=T_{10}=\rme^{p} \rho \left(\frac{\partial \Gamma }{\partial p} + z \frac{\partial \Gamma }{\partial z} \right) + \rme^{2p} z \delta \rho, \\ 
T_{11}=2 \rho z \rme^{p} \frac{\partial \Gamma}{ \partial p} + \rme^{2p} z^2 \delta \rho. 
\end{eqnarray*}
The vector $T_{A}$ is given by $T_{A}=(\rme^{p} \rho \gamma(z,p), \rme^{p} \rho \gamma(z,p) z)$ and the gauge invariant scalars both vanish, $T^2=0$, $T^3=0$. 

\section{Reduction of the Perturbation Equations and the Main Existence Theorem}
\label{sec:reduction}
In what follows, we omit the exact form of various matrices and vectors if they appear in versions of the system of perturbation equations which we do not use; relevant terms are included in Appendix A as indicated. As we have imposed self-similarity on this spacetime, a natural set of coordinates to present the linearised Einstein equations in is $x^{\mu}=(z, p, \theta, \phi)$, where $z=-t/r$ is the similarity coordinate and $p=\ln(r)$. The metric in these coordinates takes the form 
\begin{equation*}
ds^2=\rme^{2p} (-dz^2+((S-z \dot{S})^2-z^2)dp^2-2zdzdp +S^2 d \Omega^2). 
\end{equation*}
Our initial full set of perturbation equations consists of both components of (\ref{linein2}), the $t-p$ and $p-p$ components of (\ref{linein1}) and the equation of motion (\ref{Gaeqn}) for $\Gamma(z,p)$. We now discuss the series of simplifications which allows us to make this choice, before stating the initial six dimensional system of mixed evolution and constraint equations \footnote[1]{To download a \textit{Mathematica} notebook containing these calculations, go to www.student.dcu.ie/$\sim$duffye27. }. 
\begin{enumerate}[(1)]  
\item We note that since $T^2=0$ in this spacetime, by (\ref{linein4}) the metric perturbation tensor $k_{AB}$ is trace-free. We use this to eliminate one component of $k_{AB}$. 
\item Since $T^2$ vanishes, (\ref{linein4}) implies that the metric perturbation $k_{AB}$ is trace-free. Additionally, given that the scalars $Q$ and and $T^3$ both vanish in the LTB case, one can show that (\ref{linein2}) is identically satisfied, assuming that the background Einstein equations, (\ref{linein3}) and (\ref{secons2}) all hold. We will therefore use (\ref{secons2}) in preference to (\ref{linein2}).
\item We note that the $t-t$ component of (\ref{linein1}) gives us a relation for the perturbation of the dust velocity $\delta \rho$ in terms of the velocity and metric perturbation. We use this equation to eliminate $\delta \rho$ from the system. 
\item The variable $k$ is the only variable which appears at second order in derivatives in the resulting system. We therefore introduce a first order reduction by letting $u(z,p)=k(z,p)$, $v(z,p)=\partial k / \partial p$ and $w(z,p)=\dot{k}$, where $\cdot = \partial / \partial z$. 
\item The resulting system of equations consists of both components of (\ref{linein2}), the $t-p$ and $p-p$ components of (\ref{linein1}) and the equation of motion (\ref{Gaeqn}) for $\Gamma(z,p)$. To this we append the auxiliary equations $\dot{k}=w$ and $\dot{v}=\partial w / \partial p$, which makes a total of seven equations. 
\item This results in a first order system of seven equations for six variables, which we combine into a vector 
\begin{equation*}
\vec{X}=(\alpha(z,p), \beta(z, p), u(z,p), v(z,p), w(z,p), \Gamma(z,p))^T \in \mathbb{R}^6. 
\end{equation*}
Here $\alpha(z,p)$ and $\beta(z,p)$ are the independent components  of the tensor $k_{AB}$ in $(z,p)$ coordinates
\[ k_{AB}= \left( \begin{array}{cc} 
\alpha(z,p) & \beta(z,p)  \\
\beta(z,p) & \delta(z,p) \\
\end{array} \right),\] 
and since $T^2=0$, (\ref{linein4}) implies that $\delta = 2z\beta(z,p) + (S(S-2z\dot{S})+z^2(-1+\dot{S}^2)) \alpha(z,p)$. Finally $\Gamma(z,p)$ is the velocity potential given in (\ref{delu}). 
\item We can write the system of equations in a more compact form as 
\begin{equation} \label{matform}
M(z) \frac{\partial \vec{X}}{\partial z} + N(z) \frac{\partial \vec{X}}{\partial p} + O(z) \vec{X} = \vec{S}_{7}(z,p), 
\end{equation}
for $7 \times 6$ matrices $M(z)$, $N(z)$ and $O(z)$ and a 7 dimensional source vector $\vec{S}_{7}$. The dimensions of this system suggest that it may be possible to rewrite it as a 6 dimensional system of evolution equations with a constraint. To identify the constraint, we look for linear combinations of the rows of $M(z)$ which add to give zero. This corresponds to a linear combination of the equations in (\ref{matform}) which has no time derivatives, that is, a constraint. We call this the Einstein constraint. 
\item Having identified this constraint, we construct a new system consisting of six of the equations of the original system, with the constraint added to them. This is fully equivalent to the original seven equation system. We can write this system in a similar manner to (\ref{matform}), 
\begin{equation} \label{matform2} 
P(z) \frac{\partial \vec{X}}{\partial z} + Q(z) \frac{\partial \vec{X}}{\partial p} + R(z) \vec{X} = \vec{\Sigma}_{6}(z,p), 
\end{equation}
for $6 \times 6$ matrices $P(z)$, $Q(z)$ and $R(z)$ and a source vector $\vec{\Sigma}_{6}$. We note that in addition to the Einstein constraint, we have the trivial constraint $\partial k / \partial p=v(z,p)$. 
\item Having identified the two constraints, we would like to use them to eliminate two variables from the system and thus reduce the number of variables from six to four. To do this, we need to diagonalise the system. We multiply through (\ref{matform2}) by $P^{-1}$ and find the Jordan canonical form $\tilde{T}$ of the matrix coefficient of $\partial \vec{X} / \partial p$, $T:=P^{-1}Q$. We also identify the similarity matrix $S$ such that $\tilde{T}=S^{-1}TS$. 
\item To diagonalise the system, we let $\vec{Y}=S \vec{X}$, where $\mbox{det}(S) \neq 0$. $\vec{Y}$ obeys the equation 
\begin{equation}
\label{Yeom}  
\frac{\partial \vec{Y}}{\partial z} + \tilde{T}(z) \frac{\partial \vec{Y}}{\partial p} + \tilde{R}(z) \vec{Y} = \vec{\sigma}_{6}(z,p), 
\end{equation}
where $\tilde{T}=S^{-1}TS$, $\tilde{R}=\dot{S}+S^{-1}P^{-1}RS$ and $\vec{\sigma}_{6}=S^{-1}P^{-1}\vec{\Sigma}_{6}$. The matrix $\tilde{T}$ is in Jordan form but is not diagonal. 
\item In terms of the components of $\vec{X}$, the components of $\vec{Y} \in \mathbb{R}^6$ are
\begin{eqnarray*} 
&Y_{1}(z,p)=-\alpha(z,p) - 8 \pi z q(z) \Gamma(z,p) + v(z,p) \\ 
& \qquad \qquad \qquad \qquad \qquad \qquad - z w(z,p) + \frac{\dot{S}(z)}{S(z)} \beta(z,p),  \\ 
&Y_{2}(z,p)=8 \pi q(z) \Gamma(z,p), \\ 
&Y_{3}(z,p)=u(z,p), \\
&Y_{4}(z,p)=\frac{-\beta(z,p) + S(8 \pi q(z) \Gamma(z,p) + w(z,p))(S-z \dot{S}) }{S(S-z\dot{S})} \\ 
& \qquad \qquad \qquad \qquad \qquad \qquad + \frac{(z+\dot{S}(S-z \dot{S}))\alpha(z,p)}{S(S-z\dot{S})}, \\ 
&Y_{5}(z,p)=\frac{(1+\dot{S})(\beta(z,p)+(-z-S+z \dot{S}) \alpha(z,p))}{2S(S-z \dot{S})}, \\
&Y_{6}(z,p)=\frac{(-1+\dot{S})(-\beta(z,p) + (z-S+z\dot{S})\alpha(z,p))}{2S(S-z \dot{S})},
\end{eqnarray*}
where $S(z)$ is the radial function and $q(z)$ is the density given by (\ref{qform}). 
\item The Einstein constraint can be given in terms of $\vec{Y}$ as   
\begin{eqnarray} 
\label{nontrivcon}
\fl c_{1}(z) Y_{1}(z,p) + c_{2}(z) Y_{2}(z,p) + c_{3}(z) Y_{3}(z,p) +c_{4}(z)Y_{4}(z,p) \\  \nonumber
\fl \qquad \qquad \qquad + c_{5}(z) Y_{5}(z,p) + c_{6}(z) Y_{6}(z,p) + c_{7}(z) \frac{\partial Y_{4}}{\partial p}(z,p)+c_{8}(z)g(p) = 0.
\end{eqnarray}
where the coefficients are not important. The $g(p)$ term comes from the source term $\vec{\Sigma}_{6}$ in (\ref{matform}), which could be written as $\vec{\Sigma}_{6}=\vec{b}(z)g(p)$ where the form of $\vec{b}(z)$ is not needed. 
\item The trivial constraint $\partial k / \partial p=v(z,p)$ can be stated in terms of $\vec{Y}$ as   
\begin{eqnarray} 
\label{trivcon}
\fl Y_{1}(z,p)+z Y_{4}(z,p)+\left(z -S +z \dot{S} \right)Y_{5}(z,p) \\ \nonumber 
\qquad \qquad \qquad \qquad +\left( z+S-z\dot{S} \right) Y_{6}(z,p)-\frac{\partial Y_{3}}{\partial p}(z,p)=0. 
\end{eqnarray}
\item 
\begin{lemma} \label{Lem1}
Suppose that (\ref{nontrivcon}) is satisfied on an initial surface $z=z_{0}$. Then assuming that (\ref{Yeom}) holds, (\ref{nontrivcon}) will be satisfied on all surfaces $z \in (z_{c}, z_{0}]$. 
\end{lemma}
\begin{proof} 
This is a straightforward but lengthy calculation which was carried out using computer algebra. 
\hfill$\square$
\end{proof}
This lemma indicates that the constraint (\ref{nontrivcon}) is propagated by the system. 
\item The trivial constraint (\ref{trivcon}) is also propagated in the sense of Lemma \ref{Lem1}. The existence of these two constraints suggests that the true number of free variables in this system is four. We therefore aim to reduce this system to a free evolution system of four variables by using the two constraints to eliminate two variables. This is carried out in such a way as to keep the system always first order in all variables. 
\item We carry out this reduction in two steps, by first passing to a five dimensional system (which is still a mixed evolution-constraint system) by solving the trivial constraint for one variable, and then passing to a four dimensional system by solving the non-trivial constraint for another variable. The advantage to carrying out the reduction in this manner is that the five dimensional system is symmetric hyperbolic, which will be useful in what follows. 
\item The five dimensional system is obtained by solving the trivial constraint (\ref{trivcon}) for the variable $Y_{1}(z,p)$. We eliminate this variable and reduce the system to five variables, with a state vector $\vec{w} \in \mathbb{R}^5$. Again we would like to put the system in a diagonalised form. We do this as before by calculating the Jordan canonical form of the matrix coefficient $H(z)$ of $\partial \vec{w} / \partial p$ and letting $\vec{u}=S \vec{w} \in \mathbb{R}^5$, where $S$ is the similarity matrix arising from the transformation of $H(z)$ into Jordan canonical form. In terms of $\vec{Y}$, the new variables are given by
\begin{eqnarray*} 
u_{1}(z,p)=f_{1}(z) Y_{3}(z,p), \qquad \qquad \qquad  &u_{2}(z,p)=Y_{2}(z,p), \\ 
u_{3}(z,p)= Y_{4}(z,p) + f_{2}(z) Y_{3}(z,p), & u_{4}(z,p)=Y_{5}(z,p) + f_{3}(z) Y_{3}(z,p), \\ 
u_{5}(z,p)=Y_{6}(z,p) + f_{4}(z) Y_{3}(z,p),
\end{eqnarray*}
where 
\begin{eqnarray*}
&f_{1}(z)=\frac{1-\dot{S}(z)}{2S(z)}, \qquad f_{2}(z)=\frac{-z\dot{S}(z)^2 + S(z)(\dot{S}(z)+z\ddot{S}(z))}{ S(z)(S(z)-z \dot{S}(z))}, \\
&f_{3}(z)=\frac{1+\dot{S}(z)}{2 S(z)}, \qquad \qquad \qquad \qquad  f_{4}(z)=\frac{-1+\dot{S}(z)}{2 S(z)}.
\end{eqnarray*}
\item The five dimensional system obeys an equation of motion of the form
\begin{equation} \label{5deom}
\frac{\partial \vec{u}}{\partial z} + \tilde{A}(z) \frac{\partial \vec{u}}{\partial p} + \tilde{C}(z) \vec{u} = \vec{\Sigma}_{5}(z,p). 
\end{equation}
We shall use this form of the system in the next two sections chiefly because it has the useful property that it is symmetric hyperbolic. This means that the matrix $\tilde{A}$,  
\begin{equation*}
\tilde{A}=\left(
\begin{array}{ccccc}
0 & 0 & 0 & 0 & 0 \\
0 & 0 & 0 & 0 & 0 \\
0 & 0 & -\frac{1}{z} & 0 & 0\\
0 & 0 & 0 & -(z-S+z\dot{S})^{-1} & 0 \\
0 & 0 & 0 & 0 & -(z+S-z \dot{S})^{-1}
\end{array}
\right), 
\end{equation*}
is symmetric. The matrix $\tilde{C}$ and the source term $\vec{\Sigma}_{5}$ are given in Appendix A. 
\item In terms of these new variables, the Einstein constraint becomes 
\begin{eqnarray} \label{mnontriv}
\fl g_{1}(z) u_{1}(z,p) + g_{2}(z) u_{2}(z,p) + g_{3}(z) u_{3}(z,p) + g_{4}(z) u_{4}(z,p) \\ \nonumber
 \qquad \qquad \qquad + g_{5}(z) u_{5}(z,p) + g_{6}(z) \frac{\partial u_{3}}{\partial p}(z,p) + g_{7}(z) g(p) =0, 
\end{eqnarray}
where the coefficients $g_{i}(z)$, $i=1,\ldots ,7$ are listed in Appendix A, and $g(p)$ is an initial data function. We note that $g_{5}(z)$ vanishes on the Cauchy horizon. 
\item In order to eliminate one more variable, we solve (\ref{mnontriv}) for $u_{2}(z,p)$. As before, we then put the new system in Jordan canonical form by writing it in terms of the vector $\vec{k} \in \mathbb{R}^4$, where in terms of $\vec{u}(z,p)$, the new variables are
\begin{eqnarray*}
k_{1}(z,p)=u_{1}(z,p), \qquad \qquad k_{2}(z,p)=\frac{u_{3}(z,p)}{f(z)}, \\ 
k_{3}(z,p)=u_{4}(z,p), \qquad \qquad k_{4}(z,p)=u_{5}(z,p), 
\end{eqnarray*}
where 
\begin{eqnarray*}
\fl f(z)=\frac{12 (3+2 a z) (2 a-3 h(z))^2 (2 a+3 h(z))}{(3+a z)(16 a^4+108 a^2 z-81 h(z)-48 a^3 h(z)-27 a (-4+3 z h(z)))},  
\end{eqnarray*}
and $h(z)=(1+az)^{1/3}$. 
\item $\vec{k}$ obeys the differential equation
\begin{equation} 
\label{4deqnofmotion}
\frac{\partial \vec{k}}{\partial z} + E(z) \frac{\partial \vec{k}}{\partial p} + B(z) \vec{k} = \vec{\Sigma}_{4}(z,p),  
\end{equation}
where $E(z)$ is given in Appendix A and we omit $B(z)$ and $\vec{\Sigma}_{4}$. This system is a free evolution system in the sense that there are no further constraints which must be obeyed by these variables. The system cannot be reduced to any simpler form than this. However, the matrix $E(z)$ is not symmetrizable, which implies that this system is not symmetric hyperbolic. This is why we choose to work with the five dimensional system (\ref{5deom}) and the Einstein constraint (\ref{mnontriv}). 
\end{enumerate}

We will slightly rewrite the five dimensional system (\ref{5deom}) as
\begin{equation} \label{5dsys}
t\frac{\partial \vec{u}}{\partial t} + A(t) \frac{\partial \vec{u}}{\partial p} + C(t) \vec{u} = \vec{\Sigma}(t,p)
\end{equation}
where now $t=z-z_{c}$, so that $t=0$ is the Cauchy horizon. In terms of the coefficient matrices and source in (\ref{5deom}), $A(t)=t \tilde{A}(z)$, $C(t)=t \tilde{C}(z)$ and $\vec{\Sigma}(t,p)=t \vec{\Sigma}_{5}(z,p)$. We briefly list here the most important properties of this system.
\begin{itemize}
\item $\vec{u}(t,p)$ is a five dimensional vector, whose components are linear combinations of the components of the gauge invariant metric and matter perturbations. 
\item $A(t)$ and $C(t)$ are five-by-five matrices. We note that $A(t)=t \tilde{A}(z)$ and the matrix $\tilde{A}(z)$ contains a factor of $h^{-1}(z)$ in the $(5,5)$ component. Here $h(z):=z+S-z \dot{S}=z-f(z)$, where $S(z)$ is the radial function (see (\ref{Sdef})) and $f(z)=-S+z\dot{S}$. $f(z_{c})=z_{c}$ so that $h(z)$ vanishes on the Cauchy horizon. If we Taylor expand $h(z)=z-z_{c}-\dot{f}(z_{c})(z-z_{c})+O((z-z_{c})^2) = t-\dot{f}(z_{c})t + O(t^2)$, then we can see that $t h^{-1}(z)$ is analytic at the Cauchy horizon where $t=0$. This in turn implies that $A(t)=t \tilde{A}(z)$ is analytic at the Cauchy horizon. Similar remarks apply to the matrix $C(t)=t \tilde{C}(z)$, since the fifth row of $\tilde{C}(z)$ contains $h^{-1}(z)$ factors. Similarly, the fifth component of $\vec{\Sigma}_{5}$ contains a factor of $h^{-1}(z)$. So overall, $A(t)$, $C(t)$ and $\vec{\Sigma}(t,p)$ are analytic for $t \geq 0$. 
\item $A(t)$ is diagonal, whereas $C(t)$ is not. The first four rows of $A(t)$, $C(t)$ and $\vec{\Sigma}(t,p)$ are $O(t)$ as $t \rightarrow 0$, while the last row of each is $O(1)$ as $t \rightarrow 0$. 
\item The source $\vec{\Sigma}$ is separable and we can write it as $\vec{\Sigma}(z,p)=\vec{h}(t) g(p)$, where $\vec{h}(t)$ is an analytic vector valued function of $t$ and $g(p)$ is an initial data term. 
\item $g(p)$ represents the perturbation of the dust velocity and is a free initial data function. The results below follow for $g \in C_{0}^{\infty}(\mathbb{R}, \mathbb{R})$ which we assume henceforth. 
\end{itemize}
In what follows, $t_{1}$  is the initial data surface and $0 \leq t \leq t_{1}$, where $t=0$ is the Cauchy horizon, so the Cauchy horizon is approached in the direction of decreasing $t$. 

\begin{theorem} 
\label{Thm1}
The IVP consisting of the system (\ref{5dsys}) along with the initial data 
\begin{eqnarray*}
\vec{u} \bigg|_{t_{1}}=\vec{f}(p)
\end{eqnarray*}
\end{theorem}
where $\vec{f} \in C_{0}^{\infty}(\mathbb{R}, \mathbb{R}^5)$, possesses a unique solution $\vec{u}(t, p)$, $\vec{u} \in \textbf{C}^{\infty}(\mathbb{R}\times(0, t_{1}], \mathbb{R}^5)$. For all $t \in (0, t_{1}], $ $\vec{u}(t, \cdot): \mathbb{R} \rightarrow \mathbb{R}^5$ has compact support. 

\begin{proof}  This is a standard result from the theory of symmetric hyperbolic systems, see Chapter 12 of \cite{McOwen}.
\hfill$\square$
\end{proof}
We note that since the constraint is propagated by the five dimensional system (see Lemma \ref{Lem1}), a choice of smooth and compactly supported initial data for the components $u_{1}(z,p)$,  $u_{3}(z,p)$,  $u_{4}(z,p)$ and  $u_{5}(z,p)$ is sufficient to ensure that $u_{2}(z,p)$ as given by the constraint (\ref{mnontriv}) is also smooth and compactly supported. Therefore, this theorem also provides sufficient conditions for the existence of unique solutions to the four dimensional free evolution system. 

\begin{remark} {\em In Section \ref{sec:chdivbehaviour} we will require solutions $\vec{u}(t,p)$ in $L^1(\mathbb{R}, \mathbb{R}^5)$ for a choice of initial data $\vec{f} \in L^1(\mathbb{R}, \mathbb{R}^5)$. It follows immediately from Theorem \ref{Thm1} by the density of $C_{0}^{\infty}$ in $L^1$ that for $0 < t \leq t_{1}$, $\vec{u}(\cdot,p) \in L^1(\mathbb{R}, \mathbb{R}^5)$. To show that we can extend our choice of initial data to $L^1$, we require a bound on $\vec{u}$, which is established in the following lemma. For this lemma, we will need Gr\"{o}nwall's inequality \cite{Codding}, which states that for continuous functions $\phi(t)$, $\psi(t)$  and $\chi(t)$, if 
\begin{equation*}
\phi(t) \leq \psi(t)+\int_{a}^{t} \chi(s) \phi(s) \, ds, 
\end{equation*}
then 
\begin{equation*}
\phi(t) \leq \psi(t) + \int_{a}^{t} \chi(s) \psi(s) \exp\left( \int_{s}^{t} \chi(u) \, du \right) \, ds. 
\end{equation*}
 }
\end{remark}

\begin{lemma}
\label{lemubound}
The $L^1$-norm of $\vec{u}(t,p)$ obeys the bound
\begin{equation}
\label{ubound}
|| \vec{u}(t) ||_{1} \leq c_{1}(t) || \vec{u}(t_{1}) ||_{1} + c_{2}(t) ||g||_{1}, 
\end{equation}
for $0 < t \leq t_{1}$, where $c_{1}(t)$ and $c_{2}(t)$ are continuous on $(0, t_{1}]$. 
\end{lemma}
\begin{proof}
$\vec{u}$ obeys the equation
\begin{equation} \label{5dsysagain}
t\frac{\partial \vec{u}}{\partial t} + A(t) \frac{\partial \vec{u}}{\partial p} + C(t) \vec{u} = \vec{\Sigma}(t,p)
\end{equation}
where $\vec{\Sigma} = \vec{h}(t) g(p)$ and $g(p)$ in an initial data term. Each row of (\ref{5dsysagain}) can be written as 
\begin{equation}
\label{uieqn}
t\frac{\partial u_{i}}{\partial t} + a_{i}(t) \frac{\partial u_{i}}{\partial p} + c_{i}(t) u_{i} = \Sigma_{i}(t,p) - \sum_{j=1, j \neq i}^{5} c_{ij}(t)u_{j}(t,p) = S_{i}(t,p), 
\end{equation}
for $i=1, \ldots, 5$, where $a_{i}(t)$ and $c_{i}(t)$ are the diagonal components of the matrices $A(t)$ and $C(t)$ respectively, and because $C(t)$ is not diagonal, the off-diagonal components $c_{ij}(t)$ are put into the source term $S_{i}(t,p)$. We can solve (\ref{uieqn}) using the method of characteristics. The characteristics are given by 
\begin{equation}
\label{thechars}
\frac{d p_{i}}{dt} = \frac{a_{i}(t)}{t}  \qquad  \Rightarrow \qquad  p_{i}(t) = \eta_{i}+\pi_{i}(t), 
\end{equation}
where $\pi_{i}(t)=-\int_{t}^{t_{1}} \frac{a_{i}(\tau)}{\tau} d \tau$ and $\eta_{i}=p_{i}(t_{1})$. On characteristics, (\ref{uieqn}) becomes 
\begin{equation}
\label{uionchars}
t\frac{d u_{i}}{d t}(t, p_{i}(t))  + c_{i}(t) u_{i}(t, p_{i}(t)) = S_{i}(t, p_{i}(t)). 
\end{equation}
The integrating factor for (\ref{uionchars}) is $\rme^{\xi_{i}}(t)$ where $\xi_{i}(t)=-\int_{t}^{t_{1}} \frac{c_{i}(\tau)}{\tau} d\tau $, and the solution to (\ref{uionchars}) is 
\begin{equation}
\label{usoln}
\fl u_{i}(t,p_{i}) = \rme^{-\xi_{i}}(t) u_{i}^{(0)}(p_{i} - \pi_{i}(t)) - \rme^{\xi_{i}}(t) \int_{t}^{t_{1}} \frac{\rme^{\xi_{i}}(\tau)}{\tau} S_{i}(\tau, p_{i}+\pi_{i}(\tau) - \pi_{i}(t)) d \tau. 
\end{equation}
We take the $L^1$-norm by taking an absolute value and integrating with respect to $p$; this produces 
\begin{eqnarray*}
 ||u_{i}(t)||_{1} = \rme^{-\xi_{i}}(t) ||u_{i}^{(0)}||_{1} \\ 
 \qquad \qquad + \rme^{\xi_{i}}(t) \int_{\mathbb{R}} \left( \int_{t}^{t_{1}} \frac{\rme^{\xi_{i}}(\tau)}{\tau} |S_{i}(\tau, p_{i}+\pi_{i}(\tau) - \pi_{i}(t))| d \tau \right) \, dp. 
\end{eqnarray*}
Recall Theorem \ref{Thm1} which tells us that at each $t \in (0, t_{1}]$, $u_{j} \in C_{0}^{\infty}(\mathbb{R}, \mathbb{R})$. This allows us to apply Fubini's theorem; that is, to interchange the order of the integrals above. We note that the structure of the characteristics (\ref{thechars}) indicates that evaluation of the $L^1$-norm of a function $f(t,p)$ at fixed time $t$ yields the same result as the evaluation of the $L^1$-norm of $f$ evaluated on characteristics, that is 
\begin{equation*}
\int_{\mathbb{R}} |f(\tau, p| \, dp = \int_{\mathbb{R}} |f(\tau, p_{i} + \pi_{i}(\tau) - \pi_{i}(t))| \, dp. 
\end{equation*}
So, applying Fubini's theorem and using the form of $S_{i}$ produces 
\begin{eqnarray}
\label{uil1}
\fl ||u_{i}(t)||_{1} = \rme^{-\xi_{i}}(t) ||u_{i}^{(0)}||_{1} + \\ \nonumber 
\fl \qquad \qquad  \rme^{\xi_{i}}(t)  \int_{t}^{t_{1}} \frac{\rme^{\xi_{i}}(\tau)}{\tau} \left( |h_{i}(\tau)| ||g ||_{1} + \sum_{j=1, j \neq i}^{5} |c_{ij}(\tau)| \int_{\mathbb{R}} |u_{j}| \, dp \right) d \tau. 
\end{eqnarray}
Then if we note that $||g||_{1}$ does not depend on $t$, we can write 
\begin{equation}
\fl || \vec{u}(t) ||_{1} = \sup_{i} || u_{i}(t) ||_{1} \leq d_{1}(t) || \vec{u}^{(0)} ||_{1} + d_{2}(t)||g||_{1} + d_{3}(t) \int_{t}^{t_{1}} d_{4}(\tau) ||\vec{u}(\tau)||_{1} d \tau, 
\end{equation}
where the $d_{i}(t)$ functions are the suprema of the various $t$-dependent functions which appear in (\ref{uil1}) (and their precise value is not important). Now applying Gr\"{o}nwall's inequality (with $\psi(t) = d_{1}(t) || \vec{u}^{(0)} ||_{1} + d_{2}(t)||g||_{1}$) produces 
\begin{eqnarray*}
\fl || \vec{u}(t) ||_{1} \leq d_{1}(t) || \vec{u}^{(0)} ||_{1} + d_{2}(t)||g||_{1} \\ 
\qquad \qquad  + \int_{t}^{t_{1}} d_{5}(\tau) \left( d_{1}(\tau) || \vec{u}^{(0)} ||_{1} + d_{2}(\tau)||g||_{1} \right)\exp\left( \int_{\tau}^{t_{1}} d_{5}(\tau') d \tau' \right) d \tau, 
\end{eqnarray*}
where the exact value of $d_{5}(t)$ is unimportant. Again, since $||g||_{1}$ and $||\vec{u}^{(0)}||_{1}$ do not depend on $t$, we can summarize this as 
\begin{equation*}
|| \vec{u}(t) ||_{1} \leq c_{1}(t) || \vec{u}(t_{1}) ||_{1} + c_{2}(t) ||g||_{1}, 
\end{equation*}
for some continuous functions $c_{1}(t)$ and $c_{2}(t)$ defined on $(0, t_{1}]$, whose exact value is not important. 

\hfill$\square$
\end{proof}

\begin{figure} 
\begin{center}
\includegraphics[scale=1]{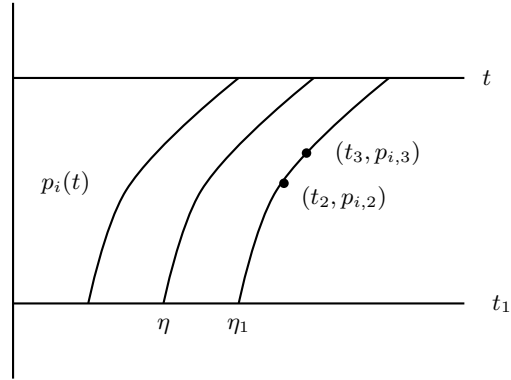}
\end{center}
\caption{Characteristics: We show here typical characteristic curves $p_{i}=p_{i}(t)$, along which the solution (\ref{usoln}) is evaluated.  
}
\label{Fig3}
\end{figure}

\begin{remark} {\em The characteristics $p_{i}(t)$ provide a $C^{1}$ foliation of the region $\Omega= \{ (t,p_{i}) : t \in (0, t_{1}], p_{i} \in \mathbb{R} \}$. A typical such foliation is shown in Figure \ref{Fig3}. For every point $q=(t_{2}, p_{i, 2}) \in \Omega$, there is a unique characteristic $C_{1}$ such that $q \in C_{1}$ which we label by $\eta_{i,2}=p_{i}(t_{2}) \bigg|_{C_{1}}$. Define the set $\Omega_{t}:=\{(t', p_{i}') \in \Omega : t'=t \}$. Then the characteristics provide a natural diffeomorphism of $\Omega_{t}$,
\begin{equation*}
p_{i, 2} \in \Omega_{t_{2}} \rightarrow p_{i, 3} \in \Omega_{t_{3}}, 
\end{equation*}
with 
\begin{equation*}
p_{i, 2} = p_{i, 3} - \int_{t_{2}}^{t_{3}} a_{i}(\tau) d \tau. 
\end{equation*}
}
\end{remark}
The fact that $\vec{u}$ obeys a bound of the form (\ref{ubound}) allows us to extend our initial data to $\vec{f}=\vec{u}(t_{1}, p) \in L^1(\mathbb{R}, \mathbb{R}^5)$ . 

\begin{theorem} 
\label{Thml1exist}
The IVP consisting of the system (\ref{5dsys}) along with the initial data 
\begin{eqnarray*}
\vec{u} \bigg|_{t_{1}}=\vec{f}(p)
\end{eqnarray*}
\end{theorem}
where $\vec{f} \in L^1(\mathbb{R}, \mathbb{R}^5)$, possesses a unique solution $\vec{u}(t, p)$, $\vec{u} \in C^{\infty}((0, t_{1}], L^1(\mathbb{R}, \mathbb{R}^5))$.  

\begin{proof}  The proof of this result relies on the bound (\ref{ubound}) and follows by a standard argument exploiting the density of $C_{0}^{\infty}(\mathbb{R}, \mathbb{R}^5)$ in $L^1(\mathbb{R}, \mathbb{R}^5)$. See Theorems 5 and 7 of \cite{scalar} for examples of such techniques and Chapter 12 of \cite{McOwen} for background details. 
\hfill$\square$
\end{proof}
Although these theorems provide for the existence of smooth or $L^1$ solutions prior to the Cauchy horizon, it gives us no information about their behaviour as they reach the horizon itself. We must therefore consider this behaviour separately. We note that this problem is rendered nontrivial by the fact that the Cauchy horizon is a singular hypersurface of (\ref{5dsys}).


\section{Behaviour of the $L^q$-Norm}
\label{sec:lqblowup}

In this section, we will use Theorem \ref{Thml1exist} to provide for the existence and uniqueness of $L^1$ solutions to (\ref{5dsys}) with a choice of $L^1$ initial data. This theorem applies for $t \in (0, t_{1}]$ only and we must consider the behaviour of $\vec{u}$ on the Cauchy horizon separately.  Our strategy in tackling this problem is as follows (see \cite{Brienunp}). We expect that any divergence which might arise in the perturbation would be in some sense (to be defined) independent of the radial coordinate, since the Cauchy horizon is a hypersurface of constant $t$, and since the coefficients of (\ref{5dsys}) are independent of $p$. Motivated by this observation, we introduce the integral of the perturbation vector with respect to the radial coordinate, which acts as a kind of ``average'' of the perturbation. This variable obeys a relatively simple system of ODEs, the solutions to which can be determined. 

Let $\vec{u}$ be the solution of (\ref{5dsys}) with $\vec{u}(t_{1}, p)=\vec{u}^{(0)}$. Then define
\begin{equation} \label{ubar}
\bar{u}(t):=\int_{\mathbb{R}} \vec{u}(t,p) \, dp,
\end{equation}
which is a kind of ``average'' of $\vec{u}(t,p)$ (note that the existence of $\bar{u}$ is guaranteed since $| \bar{u}| \leq || \vec{u}||_{1} < \infty$ since $\vec{u} \in L^1$). If we integrate with respect to $p$ through the system (\ref{5dsys}), we find that $\bar{u}$ obeys the ODE
\begin{equation} \label{ubareqn}
t \frac{d \bar{u}}{d t}= -C(t) \bar{u} + \bar{\Sigma},
\end{equation}
where $\bar{\Sigma}(t):=\int_{\mathbb{R}} \vec{\Sigma}(t,p) \, dp$. This ODE displays a regular singular point at $t=0$ (see Chapter 2 of \cite{Wasow} and Chapter 4 of \cite{Codding} for the theory of such points). We now state a theorem which gives the fundamental matrix for this system. Recall that the fundamental matrix for an ODE system is a matrix whose rows are linearly independent solutions to the ODE in question. 
\begin{theorem}
\label{Thm2}
The fundamental matrix corresponding to (\ref{ubareqn}) is 
\begin{eqnarray} \label{ubarsoln}
H(t)&=J(t) + K(t) \\ \nonumber 
    &= P(t) \, t^{-\bar{C}_{0}} + P(t) \, t^{-\bar{C}_{0}} \, \int^{t_{1}}_{t} \, P^{-1}(\tau) \, \tau^{\bar{C}_{0}-\mathbb{I}}  \bar{\Sigma} \, d \tau,
\end{eqnarray}
where $P(t)=\mathbb{I} + t P_{1} + \ldots$ is a matrix series whose coefficients can be found by a recursion relation from the Taylor expansion of $C(t)$. $\bar{C}_{0}$ is the Jordan canonical form of the zero order term in the Taylor expansion of $C(t)$. $\bar{C}_{0}$ takes the form  $\bar{C}_{0}=\mbox{diag}(0,0,0,0,c)$ where $c$ is a constant given by 
\begin{equation*}
c= \lim_{t \rightarrow 0} \, t \left(\frac{3+(S(z_{c})-z_{c})\ddot{S}(z_{c})+ \dot{S}(z_{c})(-3+z_{c} \ddot{S}(z_{c}))}{(1-\dot{S}(z_{c}))(z_{c}+S(z_{c})-z_{c}\dot{S}(z_{c}))}\right).  
\end{equation*}
$c$ depends only on $a$, and satisfies $c \in (3, + \infty)$, with 
\begin{equation*}
\lim_{a \rightarrow 0^{+}} c = 3,  \qquad \qquad \qquad \lim_{a \rightarrow a^*} c = + \infty,
\end{equation*}
where $a^*$ is the maximum value of $a$ for which a naked singularity forms. 
\end{theorem}
\begin{proof} The proof is a standard result for systems of the form (\ref{ubareqn}). See \cite{Wasow} or \cite{Codding}  for details. We have relegated the proof of the results about the behaviour of $c$ to Appendix B. 
\footnote[1]{We note that there exists a set $A$ of values of $a$ such that $c(a)$ is a natural number, that is $A=\{a \in (0, a^*) : c(a) \in \mathbb{N} \cap (3, \infty) \}$. When $c \in \mathbb{N} \cap (3, \infty)$, the fundamental matrix (\ref{ubarsoln}) will contain extra log terms. However, since this set has zero measure in the set $a \in (0, a^*)$ we will not consider it further. See \cite{Wasow} and \cite{Codding} for further details. }
\hfill$\square$
\end{proof} 
Our next task is to analyse in more detail the behaviour of this fundamental matrix. 
\subsection{Behaviour of the Fundamental Matrix}
\label{fundbeh}
We now expand about the Cauchy horizon so that $P(t) = \mathbb{I} + O(t)$ and write the homogeneous part and the particular part of $H(t)$ separately. Given the form of the matrix $C_0$ listed in Theorem \ref{Thm2}, the homogeneous part takes the form
\begin{equation*} 
J(t)=\mbox{diag}(1+O(t),1+O(t),1+O(t),1+O(t),t^{-c}+O(t^{-c+1})),
\end{equation*}
and the particular part can be written
\begin{equation*} 
K(t)=\mbox{diag}({\kappa_{1}, \kappa_{2}, \kappa_{3}, \kappa_{4}, \kappa_{5}}),
\end{equation*}
where 
\begin{equation*}
\kappa_{i}=\int^{t_{1}}_{t}  \, \tau^{-1} \bar{\Sigma}_{i}(\tau)(1+O(\tau)) \, d\tau,
\end{equation*}
for $i=1,2,3,4$ and 
\begin{equation*}
\kappa_{5} = t^{-c} \int^{t_{1}}_{t} \tau^{c-1} \, \bar{\Sigma}_{5}(\tau)(1+O(t)) \, d \tau. 
\end{equation*}
We will now examine the particular part. We recall that $\bar{\Sigma}_{i}$ is separable, so that $\bar{\Sigma}=\vec{h}(t) \, G$, where $\vec{h}(t)=(0, 0,tk_{3}(t), tk_{4}(t), k_{5}(t))$. The $k_{i}(t)$ functions are all analytic at $t=0$ and $G:=\int_{\mathbb{R}} g(p) \, dp \in \mathbb{R}$. Now the $\kappa$ terms become
\begin{eqnarray*}
\kappa_{1} = \kappa_{2} = 0, \\ 
\kappa_{j}=G k_{j}(t^*)(t_{1}-t)+O((t_{1}-t)^2),  
\end{eqnarray*}
for $j=3,4$ and 
\begin{equation*}
\kappa_{5} = t^{-c}k_{5}(t^*)G \left( \frac{t_{1}^c - t^c}{c} \right) + O(t_{1}-t)
\end{equation*}
We can see that these integrals have the same order behaviour as the corresponding homogeneous terms, that is, the $\kappa_{i}$ are $O(1)$ as $t \rightarrow 0$, for $i=1, \ldots ,4$ and $\kappa_{5}$ is $O(t^{-c})$ as $t \rightarrow 0$. 

Now since $c>0$, (\ref{ubarsoln}) shows that solutions to (\ref{ubareqn}) blow up as $t \rightarrow 0$. We now examine this divergence. 


\subsection{Blow-up of the $L^{q}$-norm}
\label{lqblowup}
We begin this analysis by determining a way to distinguish between those initial data which lead to diverging solutions to (\ref{ubareqn}) and those which do not. If we include both the homogeneous and the inhomogeneous parts, then we can label the five solutions to (\ref{ubareqn}) arising from Theorem \ref{Thm2} as 
\begin{eqnarray*}
&\bar{\phi}_{1}(z)=(1+\kappa_{1}, 0,0,0,0)^T, \\ 
&\bar{\phi}_{2}(z)=(0, 1+\kappa_{2},0,0,0)^T,\\ 
&\bar{\phi}_{3}(z)=(0, 0,1+\kappa_{3},0,0)^T, \\ 
&\bar{\phi}_{4}(z)=(0, 0,0,1+\kappa_{4},0)^T, \\
&\bar{\phi}_{5}(z)=(0,0,0,0,t^{-c}+\kappa_{5})^T,
\end{eqnarray*}
where the $\bar{\phi}_{1,2,3,4}$ are finite as $t \rightarrow 0$ and $\bar{\phi}_{5}$ is divergent. Given that (\ref{ubareqn}) has coefficients which are analytic on $(0,t_1]$, it follows that $\bar{\phi}_{1-5}$ are analytic on $(0,t_1]$. Thus these solutions provide a basis for solutions of (\ref{ubareqn}) on $(0,t_1]$. Hence given any solution $\bar{u}(t)$ of (\ref{ubareqn}), there exist constants $d_i, i=1, \ldots, 5$ such that
\[ \bar{u}(t)=\sum_{i=1}^5 d_i\bar{\phi}_i(t).\] 
Let $\mathcal{S}=L^1(\mathbb{R},\mathbb{R}^5)$. We consider solutions of (\ref{5dsys}) with initial data in $\mathcal{S}$. Given $\vec{u}^{(0)} \in \mathcal{S}$, define $\bar{u}_0=\int_{\mathbb{R}}\vec{u}^{(0)} \, dp$. We can define $d_i(\vec{u}^{(0)})$ via the existence of unique constants $d_i, i=1, \ldots, 5$ for which
\[ \bar{u}_0=\sum_{i=1}^5 d_i\bar{\phi}(t_1).\]
Define
\[ \mathcal{S}^\prime = \{\vec{u}^{(0)} \in \mathcal{S}: d_5(\vec{u}^{(0)})=0\}.\]
This set corresponds one-to-one with initial data for (\ref{5dsys}) which give rise to solutions for which the corresponding solutions of (\ref{ubareqn}) are finite as $t \to 0$. 

We define the set complement of $\mathcal{S}^\prime$ in $\mathcal{S}$ as $\mathcal{S}^{\prime\prime}=\mathcal{S}-\mathcal{S}^\prime$.

\begin{lemma}
\label{Lem2}
Given a choice of initial data $\vec{u}^{(0)} \in \mathcal{S}''$, the solution $\vec{u}$ corresponding to this data displays a blow-up of its $L^1$-norm, that is, 
\begin{equation*} 
\lim_{t \rightarrow 0} || \vec{u} ||_{1} = \infty. 
\end{equation*}
\end{lemma}
\begin{proof}
We define $\bar{u}=\int_{\mathbb{R}} \vec{u} \, dp$ where $\vec{u}$ is the solution of (\ref{5dsys}) corresponding to $\vec{u}^{(0)} \in \mathcal{S}''$. It is immediately clear from (\ref{ubarsoln}) that 
\begin{equation*} 
\lim_{z \rightarrow z_{c}} | \bar{u} | = \lim_{t \rightarrow 0} | \bar{u} | = \infty.  
\end{equation*}
Then using the definition of the $L^1$-norm of $\vec{u}$, $|| \vec{u} ||_{1}$, it follows that 
\begin{equation} \label{L1blowup} 
|| \vec{u} ||_{1} = \int_{\mathbb{R}} | \vec{u} | \, dp \geq \left| \int_{\mathbb{R}} \vec{u} \, dp \right| = | \bar{u} |, 
\end{equation}
which implies that $|| \vec{u} ||_{1} \rightarrow \infty$ as $t \rightarrow 0$. 
\hfill$\square$
\end{proof}
This divergent behaviour in the $L^1$-norm of $\vec{u}$ could be attributed to the divergence of the support of $\vec{u}$ as it spreads from the initial surface $t_{1}$, rather than to divergent behaviour in $\vec{u}$ itself (see Figure \ref{Fig2} for an illustration of this). We must therefore consider the behaviour of the spread of the support of $\vec{u}$ from the initial surface to the Cauchy horizon. In analysing the growth of the support of $\vec{u}$, it is convenient to briefly return to the self-similar coordinate $z$.  
\begin{figure} 
\begin{center}
\includegraphics[scale=1.3]{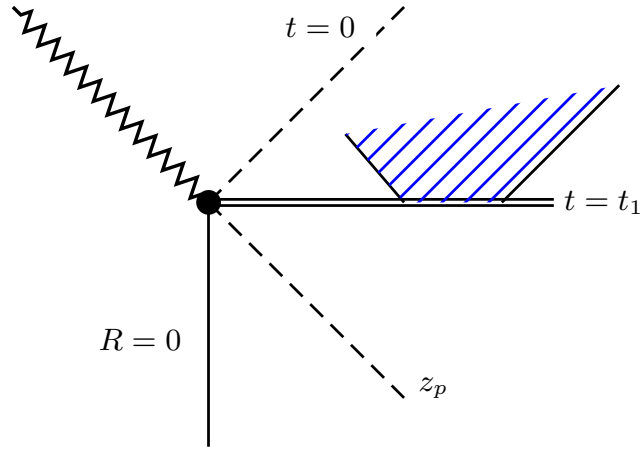}
\end{center}
\caption{The spread of the support of $\vec{u}$. We illustrate the spread of the compact support of $\vec{u}$ from the initial data surface $t_{1}$ to the Cauchy horizon. The growth of the support is bounded by in- and outgoing null rays starting from the initial data surface. 
}
\label{Fig2}
\end{figure}

\begin{lemma}
\label{lemspread}
Let the support of $\vec{u}$ be defined as $\mbox{vol}[\vec{u}](z)=p^{+}(z) - p^{-}(z)$ where 
\begin{equation*} 
p^{+}(z)=\sup_{p \in \mathbb{R}} \{p: \vec{u}(z,p) \neq 0 \},
\end{equation*}
\begin{equation*} 
p^{-}(z)=\inf_{p \in \mathbb{R}} \{p: \vec{u}(z,p) \neq 0 \}. 
\end{equation*}
Then neglecting terms which remain finite on the Cauchy horizon, the support of $\vec{u}$ grows as 
\begin{equation*} 
\mbox{vol}[\vec{u}](z) \sim -\ln|t|, \qquad \qquad t \rightarrow 0.  
\end{equation*}
\end{lemma}
\begin{proof}
Define $\mbox{vol}[\vec{u}](z)=p^{+}(z) - p^{-}(z)$ where 
\begin{equation*} 
p^{+}(z)=\sup_{p \in \mathbb{R}} \{p: \vec{u}(z,p) \neq 0 \},
\end{equation*}
\begin{equation*} 
p^{-}(z)=\inf_{p \in \mathbb{R}} \{p: \vec{u}(z,p) \neq 0 \}. 
\end{equation*}
The support of $\vec{u}$ at some time $z$, $\mbox{supp}[\vec{u}](z)$, will obey 
\begin{equation*} 
\mbox{supp}[\vec{u}](z) \subseteq [p^{-}(z), p^{+}(z)]. 
\end{equation*}
We define the $L^{q}$-norm as usual,
\begin{equation*}
|| \vec{u} ||_{q} = \left[ \int_{\mathbb{R}} |\vec{u}(z,p)|^q \, dp \right]^{1/q} = \left[ \int_{p^{-}(z)}^{p^{+}(z)} |\vec{u}(z,p)|^q \, dp \right]^{1/q},
\end{equation*}
for $1 \leq q < \infty$. $\vec{u}$ has initially compact support which implies that $\mbox{vol}[\vec{u}](z_{0}) = p^{+}(z_{0}) - p^{-}(z_{0}) < \infty$, where $z_{0}$ is the initial data surface. This initial support must grow in a causal manner; that is, the growth of $p^{\pm}(z)$ must be bounded by the in- and outgoing null directions. From the metric (\ref{metric_tr}), the in- and outgoing null directions are described by the relation $d t / d r = \pm \rme^{\nu/2} $, which in $(z,p)$ coordinates becomes
\begin{equation*} 
\frac{d z}{d p} = -(z \pm \rme^{\nu/2}),
\end{equation*}
which results in 
\begin{equation} \label{nastyint}
p^{\pm}(z)=p^{\pm}(z_{0}) + \int_{z}^{z_{0}} dz (z \pm \rme^{\nu/2})^{-1}. 
\end{equation}
In handling this integral, we first substitute $\rme^{\nu/2}=S(z) -z \dot{S}(z)$ (see (\ref{expform})), and then make the coordinate change $y=S^{1/2}(z)$. The resulting integral can be performed and results in
\begin{equation*} 
p^{+}(z)=p^{+}(z_{0}) + 3 \sum_{i=1}^4 f_{+}(y^{+}_{i}) \ln|(1+a z)^{1/3} - y^{+}_{i}|,
\end{equation*}
\begin{equation*} 
p^{-}(z)=p^{-}(z_{0}) + 3 \sum_{i=1}^4 f_{-}(y^{-}_{i}) \ln|(1+a z)^{1/3} - y^{-}_{i}|,
\end{equation*}
where 
\begin{equation*} 
f_{\pm}(k)=\frac{k^3}{-1 \pm k^2 + 4 k^3},
\end{equation*}
and $y^{\pm}_{i}$ is the $i^{th}$ root of $\pm 2 a -3 k \pm a k^3 + 3 k^4 =0 $. Therefore, the volume of the support of $\vec{u}$ grows as 
\begin{eqnarray} \label{volXgrowth}
\mbox{vol}[\vec{u}](z)=p^{+}(z_{0}) - p^{-}(z_{0}) + \\ \nonumber
\qquad 3 \sum_{i=1}^4 \left( f_{+}(y^{+}_{i}) \ln|(1+a z)^{1/3} - y^{+}_{i}| - f_{-}(y^{-}_{i}) \ln|(1+a z)^{1/3} - y^{-}_{i}| \right). 
\end{eqnarray}
Now, since in (\ref{nastyint}) there is a divergence at the Cauchy horizon when $z_{c}=-\rme^{\nu/2}(z_{c})$, the above result must contain a Cauchy horizon divergence. Using the coordinate transformation $y=S(z)^{1/2}$ it is possible to show that in terms of $y$, the Cauchy horizon is at that $y$ for which $(y^3-1)(3y-2a) + 3ay^3=0$, which is precisely where $(1+az)^{1/3} = -z(1+\frac{1}{3} az)$. When this holds, the first log term given above diverges. So in (\ref{volXgrowth}) we have one finite term describing the initially finite volume, a second term which diverges on the Cauchy horizon and a third term which is finite everywhere. So overall, if we ignore terms which remain finite as $t \rightarrow 0$, then we can describe the behaviour of the volume as 
\begin{equation*} 
\mbox{vol}[\vec{u}](z) \sim -\ln|t|, \qquad \qquad t \rightarrow 0. 
\end{equation*}
\hfill$\square$
\end{proof}
We will next need the $L^q$-embedding theorem \cite{Adams}, which we state as follows: 
\begin{theorem} \textbf{$L^q$ Embedding Theorem}: Suppose $\Omega \subseteq \mathbb{R}^n$ satisfies $vol(\Omega) = \int_{\Omega} 1 \, dx < \infty$. For $1\leq p \leq q \leq \infty$, if $u \in L^q(\Omega)$, then $u \in L^p(\Omega)$ and 
\begin{equation} 
\label{embedd}
|| u ||_{p} \leq [ \mbox{vol}(\Omega) ]^{\frac{1}{p} - \frac{1}{q}} ||u ||_{q}. 
\end{equation}
\end{theorem}
We are now in a position to show that the $L^q$-norm of $\vec{u}$ diverges. 
\begin{theorem}
\label{Thm3}
Given a choice of initial data $\vec{u}^{(0)} \in \mathcal{S}''$, the solution $\vec{u}$ of (\ref{5dsys}) corresponding to this data displays a blow-up of its $L^q$-norm for $1 \leq q \leq \infty$, that is, 
\begin{equation*} 
\lim_{t \rightarrow 0} || \vec{u} ||_{q} = \infty. 
\end{equation*}
\end{theorem}
\begin{proof} 
We set $p=1$ in (\ref{embedd}) . This produces 
\begin{equation*} || \vec{u} ||_{1} \leq [ \mbox{vol}[\vec{u}](z) ]^{1 - \frac{1}{q}} ||\vec{u} ||_{q}. 
\end{equation*}
We know from (\ref{L1blowup}) that $|| \vec{u} ||_{1} \geq |\bar{u}(z)| \sim t^{-c} $ so 
\begin{equation} \label{tcdependence}
\frac{t^{-c}}{(\mbox{vol}[\vec{u}](z))^{1-\frac{1}{q}}} \leq || \vec{u} ||_{q},
\end{equation}
Now $\lim_{t \rightarrow 0} t^{c} (\ln(t))^{1-\frac{1}{q}}$ = 0, since $c > 0$. Therefore, we can conclude that 
\begin{equation*} 
\lim_{t \rightarrow 0} || \vec{u} ||_{q} = \infty. 
\end{equation*}
So the $L^q$-norm of the solutions with initial data in $\mathcal{S}''$ blows up as $t \rightarrow 0$. 
\hfill$\square$
\end{proof}
We note that this behaviour will also hold for the four dimensional free evolution system, as the constraint is propagated (see Lemma \ref{Lem1}). We next prove two theorems which together show that this divergent behaviour is generic with respect to the initial data. In particular, we can show that the set of initial data which corresponds to solutions with divergent behaviour, $\mathcal{S}''$ , is open and dense in the set of all initial data, and that it has codimension 1 in $\mathcal{S}$. 
\begin{theorem} 
\label{Lem3}
The quotient space of $\mathcal{S}'$ in $\mathcal{S}$, $\hat{\mathcal{S}}=\mathcal{S} / \mathcal{S}'$, has codimension one in $\mathcal{S}$. 
\end{theorem}
\begin{proof} 
A quotient space $\hat{\mathcal{S}}=\mathcal{S}/\mathcal{S}^\prime$ has dimension $n$ if and only if there exist $n$ vectors $\vec{X}_{(1)},\dots,\vec{X}_{(n)}$ linearly independent relative to $\mathcal{S}^\prime$ such that for every $\vec{X}\in\mathcal{S}$, there exist unique numbers $c_1,\dots,c_n$ and a unique $\vec{X}^\prime\in\mathcal{S}^\prime$ such that $\vec{X}=\sum_{i=1}^nc_i\vec{X}_i+\vec{X}^\prime$ \cite{Janich}. So let $\vec{X}\in\mathcal{S}$ and let $\vec{X}_{(1)}$ be any element of $\hat{\mathcal{S}}$. To prove the result, we show that there is a unique value of $\alpha$ for which $\vec{X}^{(\alpha)}=\vec{X}-\alpha\vec{X}_{(1)}\in\mathcal{S}^\prime$. Integrating over the real line and exploiting earlier notation, we have
\begin{eqnarray*}
\bar{x}^{(\alpha)}&=&\bar{x}-\alpha\bar{x}_{(1)}\\
&=&\sum_{i=1}^5 (d_i(\vec{X})-\alpha d_i(\vec{X}_{(1)}))\bar{\phi}(t_1).
\end{eqnarray*}
We have $\vec{X}^{(\alpha)}\in\mathcal{S}^\prime$ if and only if $d_5(\vec{X})-\alpha d_5(\vec{X}_{(1)})=0$. Since $d_5(\vec{X}_{(1)})\neq0$ - as $\vec{X}_{(1)} \in \hat{\mathcal{S}}$ - this occurs for a unique value of $\alpha$.

\hfill$\square$
\end{proof}
\begin{theorem} 
\label{Lem4}
$\mathcal{S}''$ is dense and open in $\mathcal{S}$ in the topology induced by the $L^1$-norm. 
\end{theorem}
\begin{proof} To show that $\mathcal{S}''$ is dense in $\mathcal{S}$, we must show that for any $\vec{X} \in \mathcal{S}$ and any $\epsilon >0$, there exists some $\vec{X}'' \in \mathcal{S}''$ such that the $L^1$ distance between $\vec{X}$ and $\vec{X}''$ is less than $\epsilon$, that is,
\begin{equation}  \label{densecond}
|| \vec{X}-\vec{X}'' ||_{1} < \epsilon. 
\end{equation}
First, suppose that $\vec{X} \in \mathcal{S}''$. Then (\ref{densecond}) is trivially satisfied by taking $\vec{X}'' = \vec{X} \in \mathcal{S}''$. We therefore assume that $\vec{X} \in \mathcal{S}'$. Now consider some $\psi(p) \in C_{0}^{\infty}(\mathbb{R}, \mathbb{R})$, such that $\psi(p) \geq 0$ and $\int \psi(p) \, dp =1$. We then set
\begin{equation*} 
\vec{X}''=\vec{X}+\frac{\epsilon}{2} \psi(p) \frac{\bar{\phi}_{5}(t_{1})}{|\bar{\phi}_{5}(t_{1})|},
\end{equation*}
where $|\cdot|$ indicates the maximum vector norm in $\mathbb{R}^5$, that is $|\bar{\phi}_{5}(t_{1})|=\mbox{max}_{i} |(\bar{\phi}_{5}(t_{1}))_{i}|$ \footnote{We note that if $\bar{\phi}_{5}(t_{1}) = 0$, we can replace $\bar{\phi}_{5}$ with $\hat{\bar{\phi}}_{5} = \bar{\phi}_{5}+\sum_{i=1}^{4} c_{i} \bar{\phi}_{i}$, for some constants $c_{i}$, chosen to guarantee that $\hat{\bar{\phi}}_{5}(t_{1}) \neq 0$. This does not affect the definition of $\mathcal{S}'$ or $\mathcal{S}''$. }. Then
\begin{equation*} 
|| \vec{X}-\vec{X}'' ||_{1} =\int \frac{\epsilon}{2} \, \vec{\psi}(p) \, \left| \frac{\bar{\phi}_{5}(t_{1})}{|\bar{\phi}_{5}(t_{1})|}  \right| \, dp = \frac{\epsilon}{2} < \epsilon. 
\end{equation*}
So we can explicitly construct the $\vec{X}''$ required to satisfy (\ref{densecond}), and thus $\mathcal{S}''$ is dense in $\mathcal{S}$. 

To show that $\mathcal{S}''$ is open in $\mathcal{S}$, we must show that for all $\vec{X}'' \in \mathcal{S}''$, there exists an $\epsilon > 0$ such that $B_{\epsilon}^{1}(\vec{X}'') \subset \mathcal{S}''$, where $B_{\epsilon}^{1}(\vec{X}'')$ indicates a ball of radius $\epsilon$ in the $L^{1}$-norm centred at $\vec{X}''$. We fix $\vec{X}'' \in \mathcal{S}''$ and let $\vec{X} \in B_{\epsilon}^{1}(\vec{X}'') \subseteq \mathcal{S}$. Then 
\begin{equation}
\label{inball}
|| \vec{X}-\vec{X}''||_{1} = \int_{\mathbb{R}} |\vec{X}-\vec{X}''| \, dp < \epsilon. 
\end{equation}
There exists unique constants $c_{i}$ and $d_{i}$ (for $i=1, \ldots, 5$) such that 
\begin{eqnarray*}
\bar{x}''=d_{1} \bar{\phi}_{1} + \ldots + d_{5} \bar{\phi}_{5}, \\ 
\bar{x} = c_{1} \bar{\phi}_{1} + \ldots + c_{5} \bar{\phi}_{5}. 
\end{eqnarray*}
It follows from (\ref{inball}) that $|c_{i} - d_{i}| < \alpha_{i}\epsilon$, for some $\alpha_{i}$ depending on $\bar{\phi}_{1-5}(t_{1})$. Then by making $\epsilon$ arbitrarily small, we can make the $d_{i}$ arbitrarily close to the $c_{i}$. We know that $d_{5} \neq 0$ since $\vec{X}'' \in \mathcal{S}''$; therefore $c_{5} \neq 0$ which implies that $\vec{X} \in \mathcal{S}''$. 

\hfill$\square$
\end{proof}
These two theorems, coupled with Theorem \ref{Thm3}, suffice to show that the averaged form of the solution (\ref{ubar}) displays a generic divergence of its $L^q$-norm, where the term generic refers to the open, dense subset of the initial data which lead to this divergence. 

\begin{remark} 
\label{remark1}
{\em We note that if we define $\vec{x}:=t^c \vec{u}$ and let $\bar{x} := \int_{\mathbb{R}} \vec{x} \, dp$, then by multiplying (\ref{ubarsoln}) by a factor of $t^c$ and taking the limit $t \rightarrow 0$, we find 
\begin{equation}
\label{xbarnonzero}
\lim_{t \rightarrow 0} \bar{x} \neq 0. 
\end{equation}
This will be used in the proof of Theorem \ref{Thm10}. The results of Theorems \ref{Lem3} and \ref{Lem4} tell us that the set of initial data which gives rise to (\ref{xbarnonzero}) is open and dense in $L^1(\mathbb{R}, \mathbb{R}^5)$. }
\end{remark}

We note that from Theorem \ref{Thm3} alone we cannot conclude that the perturbation itself diverges on the Cauchy horizon. The reason for this is that one can easily imagine a function which has finite pointwise behaviour, but a diverging $L^q$-norm arising from the spatial integration in (\ref{ubar}). For example, a constant function is clearly pointwise finite, but has a diverging $L^q$-norm. In the next section, we will determine the pointwise behaviour of the perturbation as the Cauchy horizon is approached and show that it diverges with a characteristic power of $t^{-c}$. 
\section{Pointwise Behaviour at the Cauchy horizon}
\label{sec:chdivbehaviour}
So far we have determined the behaviour of the averaged perturbation $\bar{u}$. In this section, we aim to show that the vector $\vec{u}$ has behaviour similar to that of $\bar{u}$, that is, $O(1)$ behaviour in the first four components, and $O(t^{-c})$ behaviour in the last component. In this section, we will use Theorem \ref{Thm1} to provide us with smooth, compactly supported solutions $\vec{u}$ to (\ref{5dsys}), with a choice of initial data $\vec{u}(t_{1}, p) \in C_{0}^{\infty}(\mathbb{R}, \mathbb{R}^5)$. 

We begin by returning to the five dimensional symmetric hyperbolic system
\begin{equation} 
\label{5dsys2}
t\frac{\partial \vec{u}}{\partial t} + A(t) \frac{\partial \vec{u}}{\partial p} + C(t) \vec{u} = \vec{\Sigma}(t,p). 
\end{equation}
Our strategy is to work with a scaled form of $\vec{u}$, namely $\vec{x}:=t^{c} \vec{u}$. We can write an equation for $\vec{x}$ by using (\ref{5dsys2}). We find that 
\begin{equation} \label{xsys}
t\frac{\partial \vec{x}}{\partial t} + A(t) \frac{\partial \vec{x}}{\partial p} + (C(t)-c \mathbb{I}) \vec{x} = t^{c}\vec{\Sigma}(t,p). 
\end{equation}
Before presenting the results which determine the behaviour of $\vec{x}$ at the Cauchy horizon, we present a summary of various steps involved. 

\begin{itemize}
\item We begin by showing that $\vec{x}$ has a bounded energy throughout its evolution, including on the Cauchy horizon. Initially, we introduce the first energy norm, $E_{1}[\vec{x}](t)$, which is simply the $L^2$-norm of $\vec{x}$. In Theorem \ref{Thm4}, we show that this norm is bounded by a term which diverges as the Cauchy horizon is approached. 
\item We introduce a second energy norm, $E_{2}[\vec{x}](t)$ and in Theorem \ref{Thm5} show that it is bounded for $t \in [0, t^*]$, for some $t^*$ sufficiently close to the Cauchy horizon. By combining theorems \ref{Thm4} and \ref{Thm5}, we can show that $\vec{x}$ has a bounded energy up to the Cauchy horizon; see Theorem \ref{Thm6}.
\item We use this to show that $\vec{x}$ itself is bounded in Corollary \ref{Cor2}. However, there is no guarantee that $\vec{x}$ does not vanish on the Cauchy horizon. If this were to occur, then we would not be able to deduce any information about the behaviour of $\vec{u}$ from that of $\vec{x}$. 
\item We want to show that $\vec{x}$ is generically non-zero on the Cauchy horizon, for a set of non-zero measure. We can easily show that $\bar{x}:=\int_{\mathbb{R}} \vec{x} \, dp$ is non-zero at the Cauchy horizon (see Remark \ref{remark1}). If we could commute the limit $t \rightarrow 0$ with the integral, then we could show that $\int_{\mathbb{R}} \vec{x}(0, p) \, dp \neq 0$, which would be sufficient, since then $\vec{x} \neq 0$ over at least some interval on the Cauchy horizon. 
\item In order to show that we can commute the limit with the integral, we turn to the Lebesgue dominated convergence theorem, which provides conditions under which one may do this. In order to meet these conditions, we must strengthen the bound on $\vec{x}$ (Lemma \ref{Lem5}), construct a Cauchy sequence of $\vec{x}$ values in $L^1$ (Lemmas \ref{Lemderivbound} and \ref{Lem6}) and finally apply the dominated convergence theorem. 
\item So overall, we can show that $\vec{x}(0, p) \neq 0$ over at least some interval $p \in (a,b)$ on the Cauchy horizon. 
\end{itemize}
We begin with our first energy norm for $\vec{x}$. 

\subsection{Energy Bounds for $\vec{x}$}
\label{xbehaviour}
Since we expect $\vec{u}$ to diverge as $t^{-c}$, if we define $\vec{x}=t^{c} \vec{u}$, then we expect that $\vec{x}$ should have a bounded energy in the approach to the Cauchy horizon. 

\begin{theorem}
\label{Thm4}
Let $\vec{u}(t,p)$ be a solution of (\ref{5dsys}) subject to the hypotheses of Theorem \ref{Thm1}. Then there exists $t^*>0$ such that for all $t \in (0, t^{*}]$, the energy norm 
\begin{equation} \label{Mdef}
E_{1}(t):=\int_{\mathbb{R}} t^{2c} \, \vec{u} \cdot \vec{u} \, dp = \int_{\mathbb{R}} \vec{x} \cdot \vec{x} \, dp,
\end{equation}
obeys the bound
\begin{equation} \label{Mbound}
E_{1}(t) \leq \frac{\mu}{t}, 
\end{equation}
for a positive constant $\mu$ which depends only on the initial data and the background geometry. 
\end{theorem}
\begin{proof} We define $E_{1}(t)$ as in (\ref{Mdef}) and take its derivative. We substitute (\ref{xsys}) to find
\begin{equation} \label{Mderiv}
\fl \frac{d E_{1}}{d t} = \int_{\mathbb{R}} 2 t^{2c-1} (c \vec{u} \cdot \vec{u} - \vec{u} \cdot C(t) \vec{u}) - 2t^{2c-1} \vec{u} \cdot \frac{A(t)}{t} \frac{\partial \vec{u}}{ \partial p} + 2 t^{2c} \vec{u} \cdot \vec{\Sigma} \, dp. 
\end{equation}
Integrating by parts shows that the term containing $\partial \vec{u} / \partial p$ vanishes due to the compact support of $\vec{u}$ and the fact that $A(t)$ is symmetric. This leaves
\begin{equation}
\label{enrhs}
\int_{\mathbb{R}} 2t^{2c-1} (c \vec{u} \cdot \vec{u} - \vec{u} \cdot C(t) \vec{u} ) + 2t^{2c-1} \vec{u}.\vec{\Sigma} \, dp,
\end{equation}
on the right hand side of (\ref{Mderiv}). We focus on the first term in (\ref{enrhs}), and introduce the constant matrix $S$, which transforms $C_{0}$ into its Jordan canonical form and which is listed in Appendix A. We use this matrix to show that 
\begin{equation*}
c \vec{u} \cdot \vec{u} - \vec{u} \cdot C(t) \vec{u} = S^{T} \vec{u}\cdot (c \mathbb{I} - \bar{C}(t)) \vec{v},
\end{equation*}
where $\vec{v}=S^{-1} \vec{u}$ and $\bar{C}(t)=S^{-1}C(t)S$, so that $\bar{C}_0$ is the Jordan canonical form of $C_{0}$. Recall that $\bar{C}_{0}=\mbox{diag}(0,0,0,0,c)$. Now since $S^{T}=S^{T}\, S \, S^{-1}$, we find that 
\begin{equation*}
c \vec{u} \cdot \vec{u} - \vec{u} \cdot C(t) \vec{u} =\vec{v}^{T} S^{T} S \cdot (c \mathbb{I} - \bar{C}(t)) \vec{v}. 
\end{equation*}
Now for any matrix $A$, $A^{T}A$ is positive definite. We now wish to show that $\langle \vec{v}, S^{T}S(c \mathbb{I} - \bar{C}(t) \vec{v}) \rangle \geq 0$, and that equality holds iff $\vec{v}=0$. By using the form of $S$ and the matrix $\bar{C}(t=0)$ we find that the matrix $S^{T}S(c \mathbb{I} - \bar{C}(t)$ has four positive eigenvalues and one zero eigenvalue at $t=0$. This indicates that it is positive semi-definite, but could still vanish along the direction of one of the eigenvectors. Specifically, it is possible for only the fifth component of $\vec{v}$, $v_{5}$ to be non-zero and for this dot product to still vanish. However, by the continuity of the matrices here, it follows that $S^{T}S(c \mathbb{I} - \bar{C}(t)$ has four positive eignevalues for $t \in (0, t^{*})$ for some $t^{*} > 0$. Therefore, if $\langle \vec{v}, S^{T}S(c \mathbb{I} - \bar{C}(t) \vec{v}) \rangle = 0$  in this range, it follows that we must have $v_{1}=v_{2}=v_{3}=v_{4}=0$. But if these components vanish, then the only way to have $\langle \vec{v}, S^{T}S(c \mathbb{I} - \bar{C}(t) \vec{v}) \rangle = 0$  is to have $v_{5}=0$ too. We therefore conclude that if $v_{1}=v_{2}=v_{3}=v_{4}=0$, we must also have $v_{5}=0$. Therefore, $\langle \vec{v}, S^{T}S(c \mathbb{I} - \bar{C}(t) \vec{v}) \rangle = 0$ iff $\vec{v}=0$. 

So overall, we conclude that 
\begin{equation} \label{posdef}
\langle \vec{v}, S^{T}S(c \mathbb{I} - \bar{C}(t) \vec{v}) \rangle \geq 0. 
\end{equation}
for $t \in (0, t^{*})$ for some positive $t^{*}$, with equality holding iff $\vec{v}=0$. We note that if $\vec{v}=0$, then $\vec{u}=S \vec{v}=0$ too. 

We now assume that $t \in [0, t^{*})$ and using this information about the first two terms of (\ref{Mderiv}), we can conclude that 
\begin{equation*}
\frac{d E_{1}}{d t} \geq \int_{\mathbb{R}} 2t^{2c-1} \vec{u} \cdot \vec{\Sigma} \, dp. 
\end{equation*}
We now apply the Cauchy-Schwarz inequality to show that 
\begin{eqnarray*}
\frac{d E_{1}}{d t} \geq& -\int_{\mathbb{R}} t^{2c-1} \vec{u} \cdot \vec{u} \, dp - \int_{\mathbb{R}} t^{2c-1} \vec{\Sigma} \cdot \vec{\Sigma} \, dp
&=- \frac{1}{t}(E_{1}(t) -J(t)), 
\end{eqnarray*}
where $J(t):=\int_{\mathbb{R}} t^{2c} \vec{\Sigma} \cdot \vec{\Sigma}$. Integrating this from some time $t$ up to an initial time $t_{1}$ (where $0 < t \leq t_{1} \leq t^*$) will produce
\begin{equation} \label{Mbound1}
t E_{1}(t) \leq t_{1} E_{1}(t_{1}) +  \int^{t_{1}}_{t} J(\tau) \, d\tau. 
\end{equation}
We now examine the term $J(t)$. We first recall the form of the source vector $\vec{\Sigma}(t,p)=\vec{h}(t) g(p)$, where $\vec{h}(t)$ is an analytic function of $t$. $J(t)$ can therefore be written as $J(t)=t^{2c} \vec{h} \cdot \vec{h} G$, where $G:=\int_{\mathbb{R}} g(p)^2 \, dp \in \mathbb{R}$. Then the integral appearing in (\ref{Mbound1}) is 
\begin{equation*}
\int^{t_{1}}_{t} J(\tau) \, d\tau = G t_{*}^{2c} \vec{h} \cdot \vec{h} \bigg|_{t_{*}} (t_{1}-t), 
\end{equation*}
where we use the mean value theorem \cite{MVT} to put the $t$-dependent terms outside the integral and $t_{*} \in [t, t_{1}]$. The right hand side above is clearly bounded as $t \rightarrow 0$, and we can therefore conclude that 
\begin{equation*}
E_{1}(t) \leq \frac{\mu}{t}, 
\end{equation*}
for a positive constant $\mu$ which depends only on the initial data and the background geometry.  
\hfill$\square$
\end{proof}
We note that (\ref{Mdef}) is in fact the $L^2$-norm of $\vec{x}$. Since $E_{1}(t)$ is bounded, it follows that $\vec{x}(t,p) \in L^2(\mathbb{R}, \mathbb{R}^5)$ for $t \in (0, t^{*})$. We also emphasise that this bound holds only for $t>0$ and does not hold on the Cauchy horizon. We now introduce a second energy norm $E_{2}[\vec{x}](t)$ whose bound will extend to the Cauchy horizon. 

\begin{theorem}
\label{Thm5}
Let $\vec{u}$ be a solution of (\ref{5dsys}), subject to Theorem \ref{Thm1}. Define 
\begin{equation}
\label{defe2}
E_{2}[\vec{x}](t):=\int_{\mathbb{R}} \vec{x} \cdot \vec{x} + (t-1)x_{5}^{2} \, dp. 
\end{equation}
Then there exists some $t_{2}>0$ and $\mu > 0$ such that 
\begin{equation}
\label{e2bound}
E_{2}[\vec{x}](t) \leq E_{2}[\vec{x}](t_{2}) \rme^{\mu(t_{2}-t)}. 
\end{equation}
Here $\mu$ is a positive constant that depends only on the components of the background metric tensor. 
\end{theorem}
\begin{proof}
We begin by noting that the definition (\ref{defe2}) is equivalent to 
\begin{equation*}
E_{2}[\vec{x}](t):=\int_{\mathbb{R}} x_{1}^2+x_{2}^2+x_{3}^2+x_{4}^2+ tx_{5}^{2} \, dp. 
\end{equation*}
The factor of $t$ is intended to control the behaviour of $x_{5}$ as the Cauchy horizon is approached. In what follows, we will denote the $(5, 5)$ component of the matrix $A(t)$ as $a_{5}(t)$, the $(5, 5)$ component of the matrix $C(t)$ as $c_{5}(t)$, and the remaining components of the fifth row of the matrix $C(t)$ will be denoted $c_{5j}(t)$, where $j=1, \ldots ,4$. We also recall that the first four rows of $C(t)$ have $O(t)$ behaviour as $t \rightarrow 0$, and the fifth row has $O(1)$ behaviour as $t \rightarrow 0$. We begin by taking a $t$-derivative of (\ref{defe2}) and substituting from (\ref{xsys}) to find 
\begin{eqnarray}
\label{e2deriv}
\fl \frac{d E_{2}}{d t} =  \int_{\mathbb{R}} 2 \vec{x} \cdot \left(  -\frac{A(t)}{t} \frac{\partial \vec{x}}{\partial p} -\frac{C(t)-c \mathbf{I}}{t} \vec{x} + t^{c-1} \vec{\Sigma}(t,p) \right) + x_{5}^{2} \\ \nonumber 
\fl \hspace{10mm} +2(t-1)\frac{x_{5}}{t} \left( -a_{5}(t)\frac{\partial x_{5}}{\partial p} - (c_{5}(t)-c) x_{5} + t^c \Sigma_{5} - \sum_{j=1}^{4} c_{5 j}(t) x_{j} \right) \, dp, 
\end{eqnarray}
where $\sum_{j=1}^{4} c_{5 j}(t) x_{j}$ appears in the source term of the $x_{5}$ equation due to the fact that the $C(t)$ matrix is not diagonal. Now consider the terms 
\begin{equation*}
-2 \vec{x} \cdot \frac{A(t)}{t} \frac{\partial \vec{x}}{\partial p} - 2(t-1)\frac{x_{5}}{t}a_{5}(t)\frac{\partial x_{5}}{\partial p}. 
\end{equation*}
This can be simplified to 
\begin{equation}
\label{vanishterms}
\fl -2\left( x_{1} \tilde{a}_{1}(t) \frac{\partial x_{1}}{\partial p}+ x_{2} \tilde{a}_{2}(t) \frac{\partial x_{2}}{\partial p}+ x_{3} \tilde{a}_{3}(t) \frac{\partial x_{4}}{\partial p}+ x_{4} \tilde{a}_{4}(t) \frac{\partial x_{4}}{\partial p} \right) - 2x_{5}a_{5}(t)\frac{\partial x_{5}}{\partial p}, 
\end{equation}
where we used the fact that $a_{i}(t)=t\tilde{a}_{i}(t)$ where $\tilde{a}_{i}(t)$ is $O(1)$ for $i=1, \ldots ,4$. After we insert (\ref{vanishterms}) into the integral in (\ref{e2deriv}), it will vanish after an integration by parts, due to the compact support of $\vec{x}$. Returning to (\ref{e2deriv}), we are left with 
\begin{eqnarray}
\label{e2deriv2}
\frac{d E_{2}}{d t} = & \int_{\mathbb{R}} 2 \vec{x} \cdot \left( -\frac{C(t)-c \mathbf{I}}{t} \vec{x} + t^{c-1} \vec{\Sigma}(t,p) \right) + x_{5}^{2}+ \\ \nonumber 
 &2(t-1)\frac{x_{5}}{t} \left(- (c_{5}(t)-c) x_{5} + t^c \Sigma_{5} - \sum_{j=1}^{4} c_{5 j}(t) x_{j} \right) \, dp. 
\end{eqnarray}
If we now consider the terms 
\begin{equation*}
-2 \vec{x} \cdot \frac{(C(t)-c \mathbf{I})}{t} \vec{x}-2(t-1)\frac{x_{5}^2}{t}  (c_{5}(t)-c) -2(t-1)\sum_{j=1}^{4} \frac{c_{5 j}(t)}{t} x_{j}x_{5},
\end{equation*}
we notice that they can be rewritten as 
\begin{equation*}
-2 \vec{x} \cdot \frac{(\tilde{C}(t)-c \mathbf{I})}{t} \vec{x}- 2(c_{5}(t)-c) x_{5}^2 -2\sum_{j=1}^{4} c_{5 j}(t) x_{j}x_{5},
\end{equation*}
where $\tilde{C}(t)$ is a matrix got by replacing the final row of $C(t)-c \mathbb{I}$ with a row of zeroes. In other words, we can write $\tilde{C}(t)$ as 
\begin{equation*}
\tilde{C}(t)=\left(
\begin{array}{ccccc}
\hspace{1mm} & \hspace{1mm} & D & \hspace{1mm} & \hspace{1mm} \\ 
0 & 0 & 0 & 0 & 0
\end{array}
\right),
\end{equation*}
where $D$ is a $4 \times 5$ $O(t)$ matrix. Equally, we could write $\tilde{C}(t)=t \bar{C}(t)$, where $\bar{C}(t)$ is an $O(1)$ matrix. The fifth row of $\bar{C}(t)$ contains only zeroes, that is, $\bar{C}_{5i}=0$ for $i=1, \ldots ,5$. 

We insert this into (\ref{e2deriv2}) and note that $2c \vec{x} \cdot \vec{x} / t$ is explicitly positive definite. We also use the Cauchy-Schwarz inequality, in the form 
\begin{equation*}
\int_{\mathbb{R}} 2 \vec{x} \cdot \vec{\Sigma} \, dp \geq - \int_{\mathbb{R}} \vec{x} \cdot \vec{x} + \vec{\Sigma} \cdot \vec{\Sigma} \, dp, 
\end{equation*}
to find that 
\begin{eqnarray*}
\fl \frac{d E_{2}}{d t} \geq \int_{\mathbb{R}} -2 \vec{x} \cdot \bar{C}(t) \vec{x} - t^{c-1} ( \vec{\Sigma} \cdot \vec{\Sigma} + \vec{x} \cdot \vec{x})   + x_{5}^{2}+ \\ \nonumber 
\qquad -x_{5}^2(c_{5}(t)-c) + t^{c-1}(t-1) (-\Sigma_{5}^2-x_{5}^2) - 2\sum_{j=1}^{4} c_{5 j}(t) x_{j}x_{5} \, dp.  
\end{eqnarray*}
Now we let $I$ equal the integrand on the right hand side above. We introduce $I_{R}=I+\mu I_{E_{2}}$, where $\mu >0$ is a constant and $I_{E_{2}}$ indicates the integrand of (\ref{defe2}). We wish to show that $I_{R}\geq 0$. 

We can write $I_{R}$ as 
\begin{eqnarray*}
\fl I_{R}= \vec{x} \cdot (-2 \bar{C}(t) + \mu \mathbb{I} - t^{c-1} \mathbb{I}) \vec{x} + (-2(c_{5}(t) - c) + 1 -(t-1)t^{c-1}  \\ \nonumber  
\qquad \qquad  + (t-1)\mu) x_{5}^2 -2\sum_{j=1}^{4} c_{5j}(t) x_{j} x_{5} -t^{c-1} \vec{\Sigma} \cdot \vec{\Sigma} - (t-1)t^{c-1}\Sigma_{5}^2. 
\end{eqnarray*}
Let $t=0$ and note that $c>1$ so that at $t=0$, $t^{c-1}=0$. Then $I_{R}$ simplifies to 
\begin{equation*}
I_{R} \bigg|_{t=0}=\vec{x} \cdot (-2 \bar{C}(t=0) + \mu \mathbb{I}) \vec{x} + ( 1 -\mu) x_{5}^2 -2\sum_{j=1}^{4} c_{5j}(t=0) x_{j} x_{5} . 
\end{equation*}
We can simplify matters by writing $I_{R} \bigg|_{t=0}=\vec{x} \cdot H \vec{x}$, where $H$ can be written as 
\begin{equation*}
H=\left(
\begin{array}{ccccc}
\mu & 0 & 0 & 0 & 0 \\ 
0 & \mu & 0 & 0 & 0 \\ 
0 & 0 & \mu & 0 & 0 \\ 
0 & 0 & 0 & \mu & 0 \\ 
0 & 0 & 0 & 0 & 0
\end{array}
\right) + K,
\end{equation*}
where $K$ is a constant matrix, independent of $\mu$, which depends on the components of the matrix $\bar{C}(t)$ evaluated at $t=0$ and $K_{55}=1$. For $I_{R}(0) > 0$, we need $H$ to be positive definite. This implies that all of the principal subdeterminants of $H$ must be non-negative. So we require 
\begin{equation*}
\mu + K_{11} >0   
\end{equation*}

\begin{equation*}
\left|
\begin{array}{cc}
\mu+K_{11} & K_{12}  \\ 
K_{21} & \mu+K_{22}   
\end{array}
\right| >0,
\end{equation*}

\begin{equation*}
\left|
\begin{array}{ccc}
\mu+K_{11} & K_{12} & K_{13}  \\ 
K_{21} & \mu+K_{22} & K_{23}   \\ 
K_{31} & K_{32} & \mu + K_{33}
\end{array}
\right| >0,
\end{equation*}

\begin{equation*}
\left|
\begin{array}{cccc}
\mu+K_{11} & K_{12} & K_{13} & K_{14}  \\ 
K_{21} & \mu+K_{22} & K_{23} & K_{24}   \\ 
K_{31} & K_{32} & \mu + K_{33} & K_{34} \\ 
K_{41} & K_{42} & K_{43} & \mu + K_{44} 
\end{array}
\right| >0,
\end{equation*}

\begin{equation*}
|H|=\left|
\begin{array}{ccccc}
\mu+K_{11} & K_{12} & K_{13} & K_{14} & K_{15}  \\ 
K_{21} & \mu+K_{22} & K_{23} & K_{24} & K_{25}   \\ 
K_{31} & K_{32} & \mu + K_{33} & K_{34} & K_{35} \\ 
K_{41} & K_{42} & K_{43} & \mu + K_{44} & K_{45} \\ 
K_{51} & K_{52} & K_{53} & K_{54} & 1
\end{array}
\right| >0. 
\end{equation*}
These conditions produce a linear equation, a quadratic with leading $\mu^2$, a cubic with leading $\mu^3$ and two quartics with leading $\mu^4$, all of which must be positive. We therefore pick a $\mu$ which is large enough to satisfy each of these conditions. 

We therefore conclude that at $t=0$, $I_{R} > 0$. Now by continuity of the coefficients of $\vec{x}$ in $I_{R}$, it follows that there exists some $t^*$ such that $I_{R} \geq 0$ in the range $t \in [0, t^*]$. We may therefore state that 
\begin{equation}
\label{derive2bound}
\frac{d E_{2}}{d t} \geq -\mu E_{2},  
\end{equation}
in this range. We now integrate (\ref{derive2bound}) starting from some initial data surface $t_{2} \in (0, t^*)$. This results in 
\begin{equation}
E_{2}[\vec{x}](t) \leq E_{2}[\vec{x}](t_{2}) \rme^{\mu (t_{2}-t)}, 
\label{bounde2}
\end{equation}
which provides the desired bound for $E_{2}[\vec{x}](t)$. 
\hfill$\square$
\end{proof}
We note that the definition of $E_{2}[\vec{x}](t)$, (\ref{defe2}), is a sum of the $L^2$-norms of $x_{i}$ for $i=1, \ldots ,4$ and the $L^2$-norm of $t^{1/2}x_{5}$. Since we can bound $E_{2}[\vec{x}](t)$ in $t \in [0, t^*]$, it follows that in this range, $x_{i} \in L^2(\mathbb{R}, \mathbb{R})$ for $i=1, \ldots ,4$ and $t^{1/2}x_{5} \in L^2(\mathbb{R}, \mathbb{R})$. 

We next combine this result with Theorem \ref{Thm4} to provide a bound on $\vec{x}$ which holds for the entire range of $t$. 


\begin{theorem} 
\label{Thm6}
Let $\vec{u}$ be a solution of (\ref{5dsys}), subject to Theorem \ref{Thm1}. Then $E_{2}[\vec{x}](t)$ is bounded by an \textit{a priori} bound for $t \in [0, t_{1}]$, that is, 
\begin{equation}
E_{2}[\vec{x}](t) \leq  \nu E_{1}[\vec{x}](t_{1}), 
\label{fullEbound}
\end{equation}
for a positive constant $\nu$. 
\end{theorem}
\begin{proof}
To prove this, we note that by definition, $E_{2}[\vec{x}](t)=E_{1}[\vec{x}](t)+(t-1)E_{1}[x_{5}](t)$. Therefore, using the bound on $E_{1}[\vec{x}]$ from Theorem \ref{Thm2} produces $E_{2}[\vec{x}](t_{2}) \leq E_{1}[\vec{x}](t_{1})$. Inserting this into (\ref{e2bound}) produces (\ref{fullEbound}), where $\nu = \rme^{\mu}(t_{2}-t)$. 
\hfill$\square$
\end{proof}

\begin{corollary} Let $\vec{u}$ be a solution of (\ref{5dsys}), subject to Theorem \ref{Thm1}. Then $\vec{x}=t^{c} \vec{u}$ is uniformly bounded in the range $t \in (0, t^{*})$. That is
\label{Cor2}
\begin{equation*}
|x_{i}| \leq \beta_{i}  
\end{equation*}
for $i=1, \ldots ,4$ and 
\begin{equation*}
|t^{1/2}x_{5}| \leq \beta_{5} 
\end{equation*}
where the $\beta_{j}$, $j=1, \ldots,5$, are constants depending on the background geometry and on the initial data. 
\end{corollary}
\begin{proof} 
We first note that one of the effects of self-similarity has been to produce a differential operator on the left hand side of (\ref{xsys}) which has only $t$-dependent coefficients. This means that the spatial derivative, $\vec{x},_{p}$ obeys the same differential equation as $\vec{x}$, but with a modified source term. It follows that if we define $E_{1}[\vec{x},_{p}](t):=\int t^{2c} \vec{u},_{p} \cdot \vec{u},_{p} \, dp$, we can bound this energy in an exactly similar manner to Theorem \ref{Thm2}. Similarly, we can bound the energy $E_{2}[\vec{x,_{p}}](t)$ using the same argument as that of Theorem \ref{Thm5}. 

Now recall Sobolev's inequality which states that 
\begin{equation*} 
|\vec{v}|^2 \leq \frac{1}{2} \int_{\mathbb{R}} |\vec{v}|^2 + |\vec{v},_{p}|^2 \, dp, 
\end{equation*}
for $\vec{v} \in C_{0}^{\infty}(\mathbb{R}, \mathbb{R}^5)$. Applying this to $\vec{x}$ and using the bound (\ref{fullEbound}) (with the corresponding bound for $\vec{x},_{p}$) produces
\begin{equation*} 
|x_{i}| \leq \beta_{i}
\end{equation*}
for $i=1, \ldots ,4$ and 
\begin{equation*} 
|t^{1/2}x_{5}| \leq \beta_{5}
\end{equation*}
where the $\beta_{j}$, $j=1, \ldots ,5$, are constants depending on the background geometry and on the initial data. 
\hfill$\square$
\end{proof}


\subsection{Behaviour of $\vec{x}$ at the Cauchy Horizon}
\label{CHxbehaviour}
Having established a bound on $\vec{x}$ through the use of energy norms, we now wish to determine the behaviour of $\vec{x}$ as $t \rightarrow 0$, that is, the behaviour on the Cauchy horizon. In particular, we must establish that $\vec{x} \neq 0$ there. To do so, we must first strengthen the bound on $\vec{x}$ by using the method of characteristics. 

\begin{lemma}
\label{Lem5}
Let $\vec{u}$ be a solution of (\ref{5dsys}), subject to Theorem \ref{Thm1}. Then $\vec{x}$ obeys the bounds 
\begin{equation}
\label{strongerbound}
|x_{i}(t,p)| \leq \gamma_{i} \, t^{1/2}, \qquad \qquad |x_{5}(t,p)| \leq \gamma_{5},  
\end{equation}
for $i=1, \ldots, 4$ and constants $\gamma_{j}$, $j=1, \ldots, 5$ which depend only on the initial data and the background geometry of the spacetime. 
\end{lemma}
\begin{proof}
We begin by considering the first four rows of (\ref{xsys}), which we write as 
\begin{equation}
t \frac{\partial x_{i}}{\partial t} + a_{i}(t) \frac{\partial x_{i}}{\partial p} + (c_{i}-c) x_{i} = S_{i}(t,p), 
\label{xieqn}
\end{equation}
where $S_{i}(t,p)=t^{c} \Sigma_{i} - \sum_{{j=1, j \ne i}}^{5} c_{ij}(t) x_{j}$. Here $a_{i}(t)$ and $c_{i}(t)$ represent each entry on the main diagonal of the matrices $A(t)$ and $C(t)$ respectively (and we note that $a_{1}(t)=a_{2}(t)=0$). Since the matrix $C(t)$ is not diagonal, the term $\sum_{{j=1, j \ne i}}^{5} c_{ij}(t) x_{j}$ represents all the off-diagonal terms, which we put into the source. 

We solve equations (\ref{xieqn}) along characteristics. The characteristics are given by $p=p_{i}(t)$ where 
\begin{equation}
\label{chars}
\frac{d p_{i}}{d t} = \frac{a_{i}(t)}{t} = \tilde{a}_{i}(t) \Rightarrow p_{i}(t) = \pi_{i}(t) + \eta_{i}, 
\end{equation}
where we use the fact that $a_{i}(t)$ is $O(t)$ as $t \rightarrow 0$, so that $a_{i}(t)=t\tilde{a}_{i}(t)$, where $\tilde{a}_{i}(t)$ is $O(1)$ as $t \rightarrow 0$. $\pi_{i}(t)=-\int_{t}^{t_{1}} \tilde{a}_{i}(\tau) d\tau$ and $\eta_{i}=p_{i}(t_{1})$. On characteristics, (\ref{xieqn}) becomes 
\begin{equation}
t \frac{d x_{i}}{d t}(t, p_{i}(t)) +(c_{i}-c) x_{i}(t, p_{i}(t)) = S_{i}(t, p_{i}(t)). 
\label{xionchars}
\end{equation}
The integrating factors for these equations are given by $\rme^{\xi_{i}}(t)$ where 
\begin{equation*}
\xi_{i}(t)=-\int_{t}^{t_{1}} \frac{c_{i}(\tau) - c}{\tau} d\tau, 
\end{equation*}
and if we Taylor expand the term inside the integral about $t=0$, we will find that $\rme^{-\xi_{i}}=(t/t_{1})^c \alpha_{i}(t)$, where $\alpha_{i}(t)$ is an $O(1)$ term containing all terms other than the zero order term from the Taylor expansion. The solution to (\ref{xionchars}) is 
\begin{eqnarray}
\label{firstsoln}
\fl x_{i}(t,p_{i}(t)) = \left( \frac{t}{t_{1}} \right)^c \frac{\alpha_{i}(t)}{\alpha_{i}(t_{1})} x_{i}^{(0)}(\eta_{i}) - t^{c}\alpha_{i}(t) \int_{t}^{t_{1}} \frac{\tau^{-c-1}}{\alpha_{i}(\tau)} S_{i}(\tau, p_{i}(\tau)) d\tau, 
\end{eqnarray}
where $x_{i}^{(0)}$ is the initial data at $t=t_{1}$. Now we fix $t \in [0, t_{1}]$ and let $p_{i} \in \mathbb{R}$. Then using (\ref{chars}) we can write (\ref{firstsoln}) as 
\begin{eqnarray}
\label{xisoln}
\fl x_{i}(t,p_{i}(t)) = \left( \frac{t}{t_{1}} \right)^c \frac{\alpha_{i}(t)}{\alpha_{i}(t_{1})} x_{i}^{(0)}(p_{i}(t)-\pi_{i}(t)) \\ \nonumber 
\qquad \qquad  - t^{c}\alpha_{i}(t) \int_{t}^{t_{1}} \frac{\tau^{-c-1}}{\alpha_{i}(\tau)} S_{i}(\tau, p_{i}(t) + \pi_{i}(\tau)-\pi_{i}(t)) d\tau, 
\end{eqnarray}
Now taking the absolute value of (\ref{xisoln}) will produce two integral terms (coming from the two terms in the source $S_{i}$), which we label $I_{i1}$ and $I_{i2}$. $I_{i1}$ is given by
\begin{equation}
I_{i1}=t^{c}|\alpha_{i}(t)| \int_{t}^{t_{1}} \frac{\tau^{-c-1}}{|\alpha_{i}(\tau)|} \tau^{c+1}|h_{i}(\tau)||g(p_{i}(t) + \pi_{i}(\tau)-\pi_{i}(t))| d\tau, 
\label{firstint}
\end{equation}
where we use the fact that $S_{i}(t,p)=t^{c} \Sigma_{i} - \sum_{{j=1, j \ne i}}^{5} c_{ij}(t) x_{j}$ and that $\Sigma_{i}=th_{i}(t)g(p)$ where the  $h_{i}(t)$ terms are $O(1)$ functions as $t \rightarrow 0$. We use the mean value theorem to evaluate this, and conclude that 
\begin{equation}
I_{i1}=t^{c} \frac{|\alpha_{i}(t)|}{|\alpha_{i}(t^*)|} |h_{i}(t^*)||g(p_{i}(t) + \pi_{i}(t^*)-\pi_{i}(t))| (t_{1}-t),.  
\label{firstint2}
\end{equation}
for $t^* \in [t_{1}, t]$. That is, 
\begin{equation}
I_{i1} \leq \mu_{i} t^{c}+O(t^{c+1}), 
\label{firstint3}
\end{equation}
where $\mu_{i}= \sup_{t \in [0, t_{1}]} \frac{|\alpha_{i}(t)|}{|\alpha_{i}(t^*)|} |h_{i}(t^*)||g(p_{i}(t) + \pi_{i}(t^*)-\pi_{i}(t))|t_{1}$. 

The second integral from (\ref{xisoln}) is given by
\begin{eqnarray}
\fl I_{i2}=t^{c}|\alpha_{i}(t)| \int_{t}^{t_{1}} \frac{\tau^{-c-1}}{|\alpha_{i}(\tau)|} \left( \sum_{{j=1, j \ne i}}^{4} |c_{ij}(\tau)| |x_{j}(\tau, p_{i}(t) + \pi_{i}(\tau)-\pi_{i}(t))| \right) \\ \nonumber
\qquad +  \frac{\tau^{-c-1}}{|\alpha_{i}(\tau)|} \left( |c_{i5}(\tau)| |x_{5}(\tau, p_{i}(t) + \pi_{i}(\tau)-\pi_{i}(t))|  \right) d\tau.
\label{secint}
\end{eqnarray}
To handle this integral, we first note that $c_{ij}(t)$ is $O(t)$, so that $c_{ij}(t)=t \tilde{c}_{ij}(t)$, where $\tilde{c}_{ij}(t)$ is $O(1)$ as $t \rightarrow 0$. We then use the bounds on $\vec{x}$ coming from Corollary \ref{Cor2} (as well as using the mean value theorem as before). This produces 
\begin{equation*}
\fl I_{i2}=t^{c} \frac{|\alpha_{i}(t)|}{|\alpha_{i}(t^*)|} \left( \sum_{{j=1, j \ne i}}^{4} |\tilde{c}_{ij}(t^*)| \beta_{j} \frac{(t_{1}^{-c+1}-t^{-c+1})}{-c+1} +   |\tilde{c}_{i5}(t^*)| \beta_{5} \frac{(t_{1}^{-c+1/2}-t^{-c+1/2})}{-c+1/2}  \right), 
\end{equation*}
for $t^* \in [t_{1}, t]$. That is, 
\begin{equation}
\label{secint2}
I_{i2} \leq \nu_{i} t^{1/2}+O(t)
\end{equation}
where $\nu_{i}=\sup_{t \in [0, t_{1}]} \frac{|\alpha_{i}(t)|}{|\alpha_{i}(t^*)|} |\tilde{c}_{i5}(t^*)| \beta_{5} (-c+1/2)^{-1}$. Combining (\ref{firstint3}), (\ref{secint2}) and (\ref{xisoln}) produces 
\begin{equation}
\label{newxibound}
|x_{i}(t,p)| \leq \gamma_{i} \, t^{1/2}, 
\end{equation}
where the $\gamma_{i}$ factors are constants depending on the initial data and the background geometry. 

Using (\ref{newxibound}), it is also possible to improve our previous bound on $x_{5}(t,p)$. The equation which $x_{5}$ obeys is
\begin{equation}
t \frac{\partial x_{5}}{\partial t} + a_{5}(t) \frac{\partial x_{5}}{\partial p} + (c_{5}-c) x_{5} = S_{5}(t,p), 
\label{x5eqn}
\end{equation}
where $S_{5}(t,p)=t^{c} \Sigma_{5} - \sum_{{j=1}}^{4} c_{5j}(t) x_{j}$. The characteristics for this equation are given by
\begin{equation*}
\frac{d p_{5}}{d t} = \frac{a_{5}(t)}{t} \Rightarrow p_{5}(t) = \pi_{5}(t) + \eta_{5}, 
\end{equation*}
where $\pi_{5}(t)=-\int_{t}^{t_{1}} \frac{{a}_{5}(\tau)}{\tau} d\tau$ and $\eta_{5}=p_{5}(t_{1})$. On characteristics, (\ref{x5eqn}) becomes 
\begin{equation}
t \frac{d x_{5}}{d t}(t, p_{5}(t))  + (c_{5}-c) x_{5}(t, p_{5}(t))  = S_{5}(t, p_{5}(t)). 
\label{x5onchars}
\end{equation}
The integrating factor for this equation is given by $\rme^{\xi_{5}}(t)$ where 
\begin{equation*}
\xi_{5}(t)=-\int_{t}^{t_{1}} \frac{c_{5}(\tau) - c}{\tau} d\tau, 
\end{equation*}
and since $c_{5}(t=0)=c$, we see that $\rme^{-\xi_{5}}(t)$ is an $O(1)$ function as $t \rightarrow 0$. The solution to (\ref{x5onchars}) is 
\begin{eqnarray}
\label{x5soln}
\fl x_{5}(t, p_{5}(t))  = \rme^{-\xi_{5}} x_{5}^{(0)}(p_{5}(t)-\pi(t))  \\ \nonumber
\fl \qquad  - \rme^{-\xi_{5}} \int^{t_{1}}_{t} \frac{\rme^{\xi_{5}}(\tau)}{\tau} \left( \tau^{c}\Sigma_{5}(\tau, p_{5}(t) + \pi_{5}(\tau)-\pi_{5}(t))  -   \sum_{j=1}^{4} c_{5 j}(\tau) x_{j}(\tau, p_{5}(t) + \pi_{5}(\tau)-\pi_{5}(t)) \right) d \tau. 
\end{eqnarray}
The integral above contains two terms, 
\begin{equation}
I_{1}=\rme^{-\xi_{5}} \int^{t_{1}}_{t} \frac{\rme^{\xi_{5}}(\tau)}{\tau} \tau^{c}\Sigma_{5}(\tau, p_{5}(t) + \pi_{5}(\tau)-\pi_{5}(t)) d \tau, 
\end{equation}
and 
\begin{equation}
I_{2}= \rme^{-\xi_{5}} \int^{t_{1}}_{t} \frac{\rme^{\xi_{5}}(\tau)}{\tau}  \left(\sum_{j=1}^{4} c_{5 j}(\tau) x_{j}(\tau, p_{5}(t) + \pi_{5}(\tau)-\pi_{5}(t)) \right) d \tau. 
\end{equation}
We recall that $\Sigma_{5}(t,p) = h_{5}(t)g(p)$ and use the mean value theorem to show that 
\begin{equation}
\label{bound1}
I_{1}=\rme^{-\xi_{5}} \rme^{\xi_{5}}(t^*) h_{5}(t^*)g(p_{5}(t) + \pi_{5}(\tau)-\pi_{5}(t)) \frac{(t_{1}^c-t^c)}{c}, 
\end{equation}
for $t^* \in [t_{1}, t]$. For the second integral, we use the mean value theorem and the bound (\ref{newxibound}) which arises from the first part of this theorem. This produces 
\begin{equation}
\label{bound2}
I_{2}= \left(2 \rme^{-\xi_{5}} \rme^{\xi_{5}}(t^*) \sum_{j=1}^{4} c_{5 j}(t^*) \gamma_{j} \right) (t_{1}^{1/2}-t^{1/2}), 
\end{equation}
for $t^* \in [t_{1}, t]$. Combining (\ref{bound1}), (\ref{bound2}) and (\ref{x5soln}) produces 
\begin{equation*}
|x_{5}(t,p)| \leq \gamma_{5}, 
\end{equation*}
where $\gamma_{5}$ is a constant depending only on the initial data and the background geometry. 

\hfill$\square$
\end{proof}
The next lemma shows that we can bound $t^{1/2}x_{i},_{t}$. We use this in Lemma \ref{Lem6} to construct a Cauchy sequence of $\vec{x}$-values in $L^1$. 

\begin{lemma}
\label{Lemderivbound}
Let $\vec{u}$ be a solution to (\ref{5dsys}), subject to Theorem \ref{Thm1}. Define $\vec{\chi}:=\partial \vec{x} / \partial t$. Then 
\begin{equation}
\label{chiibound}
|t^{1/2}\chi_{i}(t,p)| \leq \eta_{1}, 
\end{equation}
for $i=1, \ldots, 4$, where $\eta_{1}$ is a constant depending only on the background geometry and the initial data. 
\end{lemma}
\begin{proof}
Define $\vec{\chi}:=\partial \vec{x} / \partial t$. By differentiating (\ref{xsys}), we see that $\vec{\chi}$ obeys
\begin{equation}
\label{veceqn}
t \frac{\partial \vec{\chi}}{\partial t} + A(t) \frac{\partial \vec{\chi}}{\partial p} + (C(t)+(1-c)\mathbb{I}) \vec{\chi}=\vec{\sigma}(t,p),
\end{equation}
where $\vec{\sigma}=(t^c \vec{\Sigma}),_{t} - A,_{t} \frac{\partial \vec{x}}{\partial p}-C,_{t}\vec{x}$. For $i=1, \ldots, 4$ this becomes 
\begin{equation}
\label{chieqn}
\fl t \frac{\partial \chi_{i}}{\partial t} + a_{i}(t) \frac{\partial \chi_{i}}{\partial p} + (c_{i}(t)+1-c) \chi_{i}=\sigma_{i}(t,p)-\sum_{j=1, j \neq i}^{5} c_{ij}(t) \chi_{j}=S_{i}(t,p),
\end{equation}
where $a_{i}(t)$ and $c_{i}(t)$ label the diagonal elements of the $A(t)$ and $C(t)$ matrices appearing in (\ref{veceqn}) and since $C(t)$ is not diagonal, the $c_{ij}(t)$ label the off-diagonal elements which we put in the source term $S_{i}(t,p)$. We note that $a_{1}(t)=a_{2}(t)=0$. 

As in Lemma \ref{Lem5}, we solve (\ref{chieqn}) on characteristics (and these are the same characteristics which appeared in Lemma \ref{Lem5}). By a similar method to that which lead to (\ref{xisoln}), we find that the solutions to (\ref{chieqn}) can be written as 
\begin{eqnarray}
\label{chisoln}
\chi_{i}(t,p)=\alpha_{i}(t) \left(\frac{t}{t_{1}} \right)^{c-1} \chi_{i}^{(0)}(p_{i}-\pi_{i}(t)) -  \\ \nonumber 
\qquad \qquad \alpha_{i}(t)t^{c-1} \int_{t}^{t_{1}} \frac{\tau^{-c}}{\alpha_{i}(\tau)} S_{i}(\tau, p_{i} - \pi_{i}(\tau) +\pi_{i}(t)) d \tau, 
\end{eqnarray}
where $\chi_{i}^{(0)}(p_{i}-\pi_{i}(t))$ is the initial data and the integrating factor is $\rme^{-\xi_{i}}(t)=t^{c-1}t_{1}^{1-c} \alpha_{i}(t)$. As in Lemma \ref{Lem5}, we Taylor expand the integrating factors and put all higher order terms into $O(1)$ functions $\alpha_{i}(t)$. We note that $c-1 >0$. Taking an absolute value of (\ref{chisoln}) produces two integrals (coming from the two terms in $S_{i}(t,p)$), which we label 
\begin{equation}
\label{i1int}
I_{i1}(t,p)=|\alpha_{i}(t)| t^{c-1} \int_{t}^{t_{1}} \frac{\tau^{-c}}{\alpha_{i}(\tau)} |\sigma_{i}(\tau, p_{i}(\tau) +\pi_{i}(\tau)-\pi_{i}(t))| d \tau, 
\end{equation}
and 
\begin{equation}
\label{i2int}
\fl I_{i2}(t,p)=|\alpha_{i}(t)| t^{c-1} \int_{t}^{t_{1}} \frac{\tau^{-c}}{\alpha_{i}(\tau)} \sum_{j=1, j \neq i}^{5} \tau \tilde{c}_{ij}(\tau) \chi_{j}(\tau, p_{i}(\tau) +\pi_{i}(\tau)-\pi_{i}(t)) d \tau, 
\end{equation}
where we recall that $c_{ij}(t)$ is $O(t)$, so that $c_{ij}(t)=t \tilde{c}_{ij}(t)$, where $\tilde{c}_{ij}(t)$ is $O(1)$. Now to deal with $I_{i1}$, we note that $\sigma_{i}(t,p) = (t^c \Sigma_{i}),_{t} - a_{i},_{t} \frac{\partial x_{i}}{\partial p}-C,_{t}\vec{x}|_{i}$, where $C,_{t}\vec{x}|_{i}$ indicates the $i^{th}$ row of the matrix $C,_{t} \vec{x}$. Using the bounds (\ref{strongerbound}) on $\vec{x}$ (and note that $C,_{t}\vec{x}|_{i}$ includes $x_{5}$), we can see that overall $\sigma_{i}(t,p)$ is $O(1)$ in $t$. It is also $C^{\infty}$ in $t$, so we can apply the mean value theorem to find
\begin{equation*}
I_{i1}(t,p) = t^{c-1} \bigg| \frac{\alpha_{i}(t^*)\sigma_{i}(t^*, p_{i}(t^*) + \pi_{i}(t^*)-\pi_{i}(t))}{\alpha_{i}(t^*)} \bigg| \frac{(t_{1}^{-c+1}-t^{-c+1})}{(-c+1)}, 
\end{equation*}
for $t^* \in [t_{1}, t]$. We can abbreviate this by writing 
\begin{equation}
\label{firstbound}
I_{i1}(t,p) \leq \eta_{1},
\end{equation}
where $\eta_{1}$ is a constant depending on the initial data and the background geometry (it inherits this dependence from the bounds on $\vec{x}$ entering into $\sigma_{i}$). 
For $I_{i2}$, we note that from (\ref{xsys}), we can deduce that $|t \vec{\chi}|$ is $O(1)$ (actually, the only $O(1)$ term is the $a_{5}(t) \partial x_{5} / \partial p$ term) and it is also $C^{\infty}$ in $t$, so as before, we apply the mean value theorem to find
\begin{equation*}
\fl I_{i2}(t,p)= \sum_{j=1, j \neq i}^{5} \bigg| \frac{\alpha_{i}(t) \tilde{c}_{ij}(t^*)}{\alpha_{i}(t^*)} \bigg| |t^* \chi_{j}(t^*, p_{i}(t*) +\pi_{i}(t^*)-\pi_{i}(t)) | \frac{(t_{1}^{-c+2}-t^{-c+2})}{{(-c+2)}}, 
\end{equation*}
for $t^* \in [t_{1}, t]$, which we can summarise as 
\begin{equation}
\label{secondbound}
I_{i2}(t,p) \leq \eta_{2}t. 
\end{equation}
Combining (\ref{chisoln}), (\ref{firstbound}) and (\ref{secondbound}) produces 
\begin{equation}
|x_{i},_{t}(t,p)|=|\chi_{i}(t,p)| \leq \eta_{1},
\end{equation}
for $i=1, \ldots, 4$, where we neglect higher order terms. 

\hfill$\square$
\end{proof}
We next use this result to show that we can define a sequence $\vec{x}^{(n)}$ of $\vec{x}$-values which is Cauchy in $L^1(\mathbb{R}, \mathbb{R}^5)$. 

\begin{lemma}
\label{Lem6}
Let $\{t^{(n)}\}$ be a sequence of $t$-values in $(0,t_{1}]$ with $\lim_{t \rightarrow 0} t^{(n)}=0$. For each $n \geq 1$, define $\vec{x}^{(n)}(p)=\vec{x}(t^{(n)}, p)$. Then $\{\vec{x}^{(n)}\}$ is a Cauchy sequence in $L^1(\mathbb{R}, \mathbb{R}^5)$. 
\end{lemma}
\begin{proof}
We define $\vec{x}^{(n)}:=\vec{x}(t^{(n)}, p)$, where the sequence $\{ t^{(n)} \}_{n=0}^{\infty}$ tends to zero as $n \rightarrow \infty$. The mean value theorem produces
\begin{equation}
| \vec{x}(t^{(m)}, p) - \vec{x}(t^{(n)}, p) | = |\vec{x},_{t}(t^*)| |t^{(m)} - t^{(n)}|, 
\label{xcauchy}
\end{equation}
for some $t^* \in (t^{(m)}, t^{(n)})$. For $i=1, \ldots, 4$, we can use the bound from Lemma \ref{Lemderivbound} and integrate with respect to $p$ to give 
\begin{equation*}
|| x_{i}(t^{(m)}, p) - x_{i}(t^{(n)}, p) ||_{1} \leq |t^{(m)} - t^{(n)}| \eta_{1} \int_{p_{1}}^{p_{2}} \, dp , 
\end{equation*}
where $p_{1}=\mbox{max}_{p \in \mathbb{R}} (\mbox{supp}[x_{i}(t^{(m)}, p)], \mbox{supp}[x_{i}(t^{(n)}, p)])$ and  \\ 
$p_{2}=\mbox{min}_{p \in \mathbb{R}} (\mbox{supp}[x_{i}(t^{(m)}, p)], \mbox{supp}[x_{i}(t^{(n)}, p)])$. From Lemma \ref{lemspread}, we know that $\mbox{supp}[x_{i}(t,p)]$ satisfies $\mbox{supp}[x_{i}(t,p)] \sim \ln(t)+\mu$, where $\mu$ is a constant which represents terms that remain finite as $t \rightarrow 0$. Now since $0 \leq t^{(m)} \leq t^{(n)} \leq t_{1}$, the largest support is that at $t^{(m)}$, so that 
\begin{equation*}
|| x_{i}(t^{(m)}, p) - x_{i}(t^{(n)}, p) ||_{1} \leq |t^{(m)} - t^{(n)}| \eta_{1} (\ln(t^{(m)})+\mu). 
\end{equation*}
We can take the $n \rightarrow \infty$ limit above and see that $x_{i}^{(n)}$, for $i=1, \ldots, 4$, is a Cauchy sequence with respect to the $L^1$-norm.

To show that $x_{5}^{(n)}$ is a Cauchy sequence in $L^1$ we use a different tactic. As in (\ref{x5soln}), we write the solution to the $x_{5}$ equation of motion as 
\begin{equation}
\label{Fx5soln}
x_{5}(t,p)=\rme^{-\xi_{5}}(t) x_{5}^{(0)}(p - \pi_{5}(t)) - F(t,p), 
\end{equation}
where 
\begin{equation}
\label{Fdef}
F(t,p)=\rme^{-\xi_{5}}(t) \int_{t}^{t_{1}} \frac{\rme^{\xi_{5}(\tau)}}{\tau} S_{5}(\tau, p + \pi_{5}(\tau) - \pi_{5}(t)) d \tau. 
\end{equation}
Here $\rme^{\xi_{5}}(t)$ is the integrating factor, and $\xi_{5}(t)=-\int_{t}^{t_{1}} \frac{c_{5}(\tau)-c}{\tau} d \tau$, $\pi_{5}(t)=-\int_{t}^{t_{1}}\frac{a_{5}(\tau)}{\tau} d \tau$ and $c_{5}(t)$ and $a_{5}(t)$ are the $(5,5)$-components of the $C(t)$ and $A(t)$ matrices respectively. $S_{5}(t,p)$ is the source term, $S_{5}(t,p)=t^c \Sigma_{5}(t,p) - \sum_{j=1}^{5}c_{5j}(t) x_{j}(t,p)$, where $c_{5j}(t)$ are the components of the fifth row of the $C(t)$ matrix appearing in (\ref{xsys}). Using (\ref{Fx5soln}) we can write 
\begin{eqnarray}
\label{x5cauchy}
\fl |x_{5}(t^{(n)}, p)-x_{5}(t^{(m)}, p)| \leq \\ \nonumber 
\fl |\rme^{-\xi_{5}}(t^{(n)}) x_{5}^{(0)}(p-\pi_{5}(t^{(n)}))-\rme^{-\xi_{5}}(t^{(m)}) x_{5}^{(0)}(p-\pi_{5}(t^{(m)}))|+|F(t^{(n)}, p)-F(t^{(m)}, p)|. 
\end{eqnarray}
Now if we suppose that $t^{(n)}$ and $t^{(m)}$ are very close to the Cauchy horizon, $t=0$, then tracing the characteristic back to the initial data surface, we can see that it will intersect the initial data surface outside the compact support of the solution. That is, there exists some $N$ such that for $n, m \geq N$, $x_{5}^{(0)}(p-\pi_{5}(t^{(n)})) = x_{5}^{(0)}(p-\pi_{5}(t^{(m)}))=0$. 

Now for the second term in (\ref{x5cauchy}), we use the mean value theorem to show that for $t^* \in [t^{(n)}, t^{(m)}]$,  
\begin{equation}
\label{mvtf}
|F(t^{(n)}, p)-F(t^{(m)}, p)| \leq \bigg| \frac{\partial F}{\partial t}(t^*, p) \bigg| |t^{(n)}-t^{(m)}|. 
\end{equation}
We can easily calculate 
\begin{equation}
\label{Fderiv}
\frac{\partial F}{\partial t} = -\frac{d \xi_{5}}{dt} F(t,p) - \frac{S_{5}(t,p)}{t}. 
\end{equation}
Now using the bounds (\ref{strongerbound}) on $\vec{x}$ (and the fact that $\Sigma_{5}=h_{5}(t)g(p)$ where $h_{5}(t)$ is $O(1)$ as $t \rightarrow 0$) we can show that 
\begin{equation*}
 \bigg| \frac{\partial F}{\partial t}(t, p) \bigg| \leq \mu \, t^{-1/2}, 
\end{equation*}
where $\mu$ is a constant depending only on the background geometry and the initial data. Combining this with (\ref{mvtf}) produces 
\begin{equation}
\label{Fbound}
|F(t^{(n)}, p)-F(t^{(m), p})| \leq \mu (t^*)^{-1/2}|t^{(n)}-t^{(m)}|. 
\end{equation}
Returning to (\ref{x5cauchy}), assuming $n, m \geq N$ and using (\ref{Fbound}) produces 
\begin{equation*}
|x_{5}(t^{(n)}, p)-x_{5}(t^{(m)}, p)| \leq \mu (t^*)^{-1/2}|t^{(n)}-t^{(m)}|. 
\end{equation*}
Finally, we take the $L^1$-norm to find 
\begin{equation*}
||x_{5}(t^{(n)}, p)-x_{5}(t^{(m)}, p)||_{1} \leq \mu (t^*)^{-1/2}|t^{(n)}-t^{(m)}| \int_{p_{1}}^{p_{2}} \, dp ,  
\end{equation*}
where $p_{1}=\mbox{max}_{p \in \mathbb{R}} (\mbox{supp}[x_{5}(t^{(m)}, p)], \mbox{supp}[x_{5}(t^{(n)}, p)])$ and  \\ 
$p_{2}=\mbox{min}_{p \in \mathbb{R}} (\mbox{supp}[x_{5}(t^{(m)}, p)], \mbox{supp}[x_{5}(t^{(n)}, p)])$. From Lemma \ref{lemspread}, we know that $\mbox{supp}[x_{5}(t,p)]$ satisfies $\mbox{supp}[x_{i}(t,p)] \sim \ln(t)+\mu$, where $\mu$ is a constant which represents terms that remain finite as $t \rightarrow 0$. Now since $0 \leq t^{(m)} \leq t^{(n)} \leq t_{1}$, the largest support is that at $t^{(m)}$, so that 
\begin{equation*}
||x_{5}(t^{(n)}, p)-x_{5}(t^{(m)}, p)||_{1} \leq \mu (t^*)^{-1/2}|t^{(n)}-t^{(m)}| (\ln(t^{(m)}+\mu).   
\end{equation*}
and taking the limit $n \rightarrow \infty$, we see that $x_{5}^{(n)}$ is also a Cauchy sequence in $L^1$. So we can conclude that $\vec{x}^{(n)}$ is a Cauchy sequence in $L^1(\mathbb{R}, \mathbb{R}^5)$. 

\hfill$\square$
\end{proof}
Now with Lemma \ref{Lem6} in place, and since we know that $\vec{x} \in L^1(\mathbb{R}, \mathbb{R}^5)$ for $t \in (0, t_{1}]$, we can show that $\vec{x}$ does not vanish on the Cauchy horizon. To do this, we will make use of two theorems from real analysis, which we state here. The proofs of both theorems are standard; see \cite{LiebLoss} for details. 

\begin{theorem}
\label{dominated}
Let $1 \leq p \leq \infty$. Suppose $\Omega$ is some set, $\Omega \subseteq \mathbb{R}^n$, and let $f^{(i)}$, $i=1, 2, \ldots$ be a Cauchy sequence in $L^p(\Omega)$. Then there exists a unique function $f \in L^p$ such that $||f^{(i)} - f||_{p} \rightarrow 0$ as $i \rightarrow \infty$, that is, $f^{(i)}$ converges strongly in the $L^p$-norm to $f$ as $i \rightarrow \infty$. 

Furthermore, there exists a subsequence $f^{(i_{1})}, f^{(i_{2})}, \ldots $, $i_{1}<i_{2}< \ldots$ and a non-negative function $F \in L^p(\Omega)$ such that 
\begin{itemize}
\item Domination: $|f^{(i_{k})}(x)| \leq F(x)$ for all $k$ and for a dense subset of $x \in \Omega$,
\item Pointwise Convergence: $\lim_{k \rightarrow \infty} f^{(i_{k})}(x) = f(x)$ for a dense subset of $x \in \Omega$. 
\end{itemize}
\end{theorem}

\begin{theorem} 
\label{Lebesgue}
(Lebesgue's Dominated Convergence Theorem) Let $\{ f^{(i)} \}$ be a sequence of summable functions which converges to $f$ pointwise almost everywhere. If there exists a summable $F(x)$ such that $|f^{(i)}(x)| \leq F(x)$ $\forall i$, then $|f(x)| \leq F(x)$ and 
\begin{equation*}
\lim_{i \rightarrow \infty} \int_{\Omega} f^{(i)}(x) dx = \int_{\Omega} f(x) dx, 
\end{equation*}
that is, we can commute the taking of the limit with the integration. 
\end{theorem}

\begin{theorem}
\label{Thm10}
Let $\vec{u}$ be a solution to (\ref{5dsys}), subject to Theorem \ref{Thm1}. Then $\vec{x}(t,p)=t^{c}\vec{u}$ does not vanish as $t \rightarrow 0$, for a generic choice of initial data. Here the term generic refers to the open dense subset of $L^1$ initial data which lead to this result. 
\end{theorem}
\begin{proof}
Since $\vec{x} \in L^1$ for $t \in (0,t_{1}]$, the proof of this theorem follows by an application of the two theorems from analysis quoted above. Consider the sequence $\vec{x}^{(n)}:=\vec{x}(t^{(n)}, p)$, where $\{ t^{(n)} \}_{n=0}^{\infty}$ tends to zero as $n \rightarrow \infty$. In Lemma \ref{Lem6}, we showed that this sequence is Cauchy in $L^1$. Theorem \ref{dominated} therefore provides for the existence of a dominated subsequence of $\vec{x}^{(n)}$. In particular, by applying this theorem we may conclude that there exists a non-negative $H \in L^1(\mathbb{R}, \mathbb{R})$ and a unique $\vec{h} \in L^1(\mathbb{R}, \mathbb{R}^5)$ such that 
\begin{equation*}
| \vec{x}^{\, (n_{m})} | \leq H(p) \quad \forall m, 
\end{equation*}
and $||\vec{x}^{\, (n_{m})} - \vec{h}||_{1} \rightarrow 0 $ as $m \rightarrow \infty$. 
Next we apply the Lebesgue dominated convergence theorem (Theorem \ref{Lebesgue}) to the dominated subsequence $\vec{x}^{\, (n_{m})}$. This produces 
\begin{equation}
\label{commute}
\lim_{m \rightarrow \infty} \int_{\mathbb{R}} \vec{x}(t^{(n_{m})}, p) \, dp = \int_{\mathbb{R}} \vec{x}(0, p) dp, 
\end{equation}
where we identify $\vec{h} \in L^1(\mathbb{R}, \mathbb{R}^5)$ with $\vec{x}(0,p)$. \footnote[1]{That is, $\vec{x}(0,p)$ is defined on a dense subset of $\mathbb{R}$ by the second result of Theorem \ref{Lebesgue}. It suffices to take any bounded extension to ``fill in'' the definition of $\vec{x}(0,p)$ on the remaining set of zero measure. }

Now if we recall Remark \ref{remark1}, which indicated that $\lim_{t \rightarrow 0} \int_{\mathbb{R}} \vec{x} \, dp \neq 0$, we can conclude that $\lim_{m \rightarrow \infty} \int_{\mathbb{R}} \vec{x}(t^{(n_{m})}, p) \, dp \neq 0$. Combining this with (\ref{commute}) produces  
\begin{equation}
\label{xnotzero}
\int_{\mathbb{R}} \vec{x}(0, p) dp \neq 0, 
\end{equation}
which implies that there exists an open subset $(a, b)$ on the Cauchy horizon such that $\vec{x}(t, p) \neq 0$ for $p \in (a, b)$. We note that (\ref{xnotzero}) holds generically since $\lim_{t \rightarrow 0} \int_{\mathbb{R}} \vec{x} \, dp \neq 0$ for a generic (that is, open and dense in $L^1$) set of initial data (see Remark \ref{remark1}). 
\hfill$\square$
\end{proof}
We conclude that $\vec{x}:=t^{c}\vec{u}$ exists and is non-zero on the Cauchy horizon for $p \in (a,b)$, for a general choice of initial data. This in turn tells us that the perturbation $\vec{u}$ diverges in a pointwise manner at the Cauchy horizon, with a characteristic power given by $t^{-c}$. 

\begin{theorem}
\label{Thmudiv}
There exists an open and dense subset of initial data $\vec{u}^{(0)} \in L^1(\mathbb{R}, \mathbb{R}^5)$ such that the solution $\vec{u}$ of (\ref{5dsys2}) corresponding to this initial data blows up as $t \rightarrow 0$ on an open subset $p \in (a,b)$, that is 
\begin{equation}
\label{udiverges}
\lim_{t \rightarrow 0} \vec{u}(t,p) = \infty,  \qquad \qquad \forall p \in (a,b). 
\end{equation} 
\end{theorem}
\begin{proof}
It follows immediately from Theorem \ref{Thm10} that $\vec{u}$ blows up as $t \rightarrow 0$ on an open subset $p \in (a,b)$ for a choice of $C_{0}^{\infty}(\mathbb{R}, \mathbb{R}^5)$ initial data. Recall Remark \ref{remark1} which indicates that $\bar{x}= \int_{\mathbb{R}} \vec{x} \, dp \neq 0$ for a generic choice of initial data from $L^1(\mathbb{R}, \mathbb{R}^5)$. We can therefore extend the results of Theorem \ref{Thm10} to a choice of initial data taken from an open and dense subset of $L^1$. We conclude that for such initial data 
\begin{equation*}
\lim_{t \rightarrow 0} \vec{u}(t,p) = \infty,  \qquad \qquad \forall p \in (a,b). 
\end{equation*} 

\end{proof}

\section{Physical Interpretation of Variables}
\label{sec:gi}
So far, we have established the behaviour of $\vec{u}$ as it approaches the Cauchy horizon. We now wish to provide an interpretation of these results in terms of the perturbed Weyl scalars, which represent the gravitational radiation produced by the metric and matter perturbations. In this section, we will use the coordinate system $(u,v, \theta, \phi)$, where $u$ and $v$ are the in- and outgoing null coordinates (see (\ref{uvcoords}) for their definitions), rather than $(z, p, \theta, \phi)$ coordinates. This is a useful choice of coordinate system to make when considering the perturbed Weyl scalars. We follow throughout the presentation of \cite{physical}. 

For the even parity perturbations, only two of the perturbed Weyl scalars, $\delta \Psi_{0}$ and $\delta \Psi_{4}$, are identification and tetrad gauge invariant (see  \cite{physical} and \cite{szekeres}). This means that if we make a change of null tetrad, or a change of our background coordinate system, we will find that these terms are invariant under such changes. We note that $\delta \Psi_{0}$ and $\delta \Psi_{4}$ represent transverse gravitational waves propagating radially inwards and outwards.These terms are given by 
\begin{eqnarray*}
\delta \Psi_{0} = \frac{Q}{2r^2} \bar{l}^{A} \bar{l}^{B} k_{AB},  \qquad \qquad \delta \Psi_{4}=\frac{Q^{*}}{2r^2} \bar{n}^{A} \bar{n}^{B} k_{AB},
\end{eqnarray*}
where $Q$ and $Q^{*}$ are angular coefficients depending on the other vectors in the null tetrad, and on the basis constructed from the spherical harmonics. The ingoing and outgoing null vectors $\bar{l}^{A}$ and $\bar{n}^{A}$ are given in (\ref{inoutvectors}). The term 
\begin{equation}
\label{deltap}
\delta P_{-1}=|\delta \Psi_{0} \delta \Psi_{4}|^{1/2}
\end{equation}
is also invariant under spin-boosts, and therefore has a physically meaningful magnitude \cite{physical}. 


\begin{theorem}
\label{Thm12}
The perturbed Weyl scalars $\delta \Psi_{0}$ and $\delta \Psi_{4}$, as well as the scalar $\delta P_{-1}$, diverge on the Cauchy horizon.
\end{theorem}
\begin{proof}
We begin by writing the tensor $k_{AB}$ in $(u,v)$ coordinates as 
\[k_{AB}= \left( \begin{array}{cc} 
\eta(u,v) & \nu(u,v)  \\
\nu(v,v) & \lambda(v,v) \\
\end{array} \right).\]
The condition that $k_{AB}$ be tracefree results in $\nu(u,v) =0$. In $(u,v)$ coordinates, the perturbed Weyl scalars become 
\begin{eqnarray*}
\delta \Psi_{0} = \frac{Q}{2r^2 B^2} \eta(u,v), \qquad \qquad \delta \Psi_{4} = \frac{Q^{*}}{2r^2} \lambda(u,v), 
\end{eqnarray*}
where the factor of $B(u,v)$ is the same factor which appears in (\ref{inoutvectors}). Now by performing a coordinate transformation, we can write $\alpha(z,p)$ and $\beta(z,p)$ (the components of $k_{AB}$ in $(z,p)$ coordinates) in terms of $\eta(u,v)$ and $\lambda(u,v)$ and by so doing, we can find $\delta \Psi_{0}$ and $\delta \Psi_{4}$ in terms of $\alpha(z,p)$ and $\beta(z,p)$. We find
\begin{eqnarray}
\label{weylzp}
\delta \Psi_{0} = F(z,p)(\alpha(z,p) - f_{-}^{-1}(z) \beta(z,p)),  \\ \nonumber 
\delta \Psi_{4} = G(z,p)(f_{+}(z)\alpha(z,p) - \beta(z,p)),
\end{eqnarray}
where the coefficients $F$ and $G$ are given by
\begin{eqnarray*}
F(z,p)=\frac{Q}{2r^2 B^2} \frac{f_{+}^2}{u^2} \left(\frac{f_{-}}{f_{-}-f_{+}}\right), \qquad \qquad G(z,p)=\frac{Q^{*}}{2r^2} \frac{f_{-}^2}{v^2} \left(\frac{1}{f_{+}-f_{-}}\right),
\end{eqnarray*}
and we note that $F \sim r^{-4}$ and $G \sim r^{-2}$ (recall that $B(u,v) \sim r^2$).  

Now, if we retrace our steps through the first order reduction in Section \ref{sec:reduction}, we find that $u_{5}(z,p)$ contributes to $\alpha(z,p)$, $\beta(z,p)$, $k(z,p)$ and its first derivatives. In particular, the pointwise divergence of $u_{5}(z,p)$ on the Cauchy horizon produces a similar divergence in these terms. We can write $\alpha$ and $\beta$ as 
\begin{eqnarray}
\label{alpandbeta}
\fl \alpha(z, p)=\frac{S}{1-\dot{S}^2} \left(u_{4}(-1+\dot{S}) + u_{5}(1+\dot{S}) + 2u_{1}(1+\dot{S}) \right),  \\ \nonumber 
\fl \beta(z, p)=\frac{S}{1+\dot{S}^2} \left(u_{4}(-1+\dot{S})(z-S+z\dot{S}) + 2z(1+\dot{S})u_{1} + u_{5}(z-z\dot{S}+S(1+\dot{S})) \right), 
\end{eqnarray}
where $S(z)=(1+az)^{2/3}$ is the radial function. Combining (\ref{weylzp}) and (\ref{alpandbeta}) produces 
\begin{eqnarray*}
\delta \Psi_{0}=F(z,p)(\beta_{1}(z)u_{4}(z,p)+\beta_{2}(z)u_{5}(z,p)+\beta_{3}(z)u_{1}(z,p)), \\ \nonumber 
\delta \Psi_{4}=G(z,p)(\beta_{4}(z)u_{4}(z,p)+\beta_{5}(z)u_{1}(z,p)+\beta_{6}(z)u_{5}(z,p)),
\end{eqnarray*}
where 
\begin{eqnarray*}
\beta_{1}(z)=\frac{2(-1+\dot{S})S \dot{S}^2}{1-\dot{S}^4}, \\
\beta_{2}(z)=S \left(\frac{1}{1-\dot{S}} - \frac{z-z\dot{S}+S(1+\dot{S})}{(z-S+z\dot{S})(1+\dot{S}^2)} \right), \\ 
\beta_{3}(z)=2S \left( \frac{1}{1-\dot{S}} - \frac{z(1+\dot{S})}{(z-S+z \dot{S})(1+\dot{S}^2)} \right), \\ 
\beta_{4}(z)=\frac{2(-1+\dot{S})}{1-\dot{S}^4} (S+z\dot{S}(-1+\dot{S})),  \\ 
\beta_{5}(z)= \frac{2S^2(1+\dot{S})^2}{1-\dot{S}^4}, \\ 
\beta_{6}(z)=\frac{S}{1-\dot{S}^4} (-\dot{S}(1+\dot{S})(-2\dot{S}S -z +\dot{S}^2z)).  
\end{eqnarray*}
So $\delta \Psi_{0}$ and $\delta \Psi_{4}$ depend on $u_{5}$ and therefore they diverge as the Cauchy horizon is approached. Similarly, $\delta P_{-1}$ diverges on the Cauchy horizon, as it depends on $\delta \Psi_{0}$ and $\delta \Psi_{4}$.  
\hfill$\square$
\end{proof}

To construct a gauge invariant interpretation for the matter term $\Gamma(z,p)$, we note that by comparing the GS terms (\ref{m2se}) and (\ref{TAdef}) to (\ref{dustsepert}), we find that 
\begin{equation*}
T_{A}=\bar{\rho}(z,p) \, \bar{u}_{A} (\Gamma(z,p)+g(p))
\end{equation*}
which we contract with the background dust velocity $\bar{u}^{A}$ to find 
\begin{equation*} 
\bar{u}^{A} T_{A} = -\bar{\rho}(z,p) (\Gamma(z,p) + g(p))
\end{equation*}
where $\bar{u}^{A}T_{A}$ is a gauge invariant scalar. 


\section{Conclusion}
\label{sec:concl}

We have considered here the behaviour of even parity perturbations of the self-similar LTB spacetime in the case where there is a naked singularity. In particular, we have examined their evolution in the approach to the Cauchy horizon associated with the naked singularity. We first identified the fundamental system of PDEs which govern the evolution of the even parity perturbations but chose to use the five dimensional system in preference, as its symmetric hyperbolicity is a very useful property. We next considered an averaged form of the perturbation variable and showed that it displayed a generic divergence in its $L^q$-norm. The divergence occurs for all angular number $l$. Here the term generic refers to the fact that the divergence occurs when solutions evolve from a set of initial data which is open and dense in the set of all initial data. This result is not by itself sufficient to ensure that the perturbation diverges on the Cauchy horizon, as it is perfectly possible for a pointwise finite function to have a diverging $L^q$-norm. However, this result gave us a useful ansatz for investigating the asymptotic behaviour of the solution itself in the approach to the Cauchy horizon. 

To determine the pointwise behaviour of the perturbation in the approach to the Cauchy horizon, we introduced a scaled version of the state vector,  $\vec{x}:=t^{c}\vec{u}$. We used a series of energy methods which produced bounds on $\vec{x}$ and its $p$-derivatives in the approach to the Cauchy horizon. By combining these results with the formal solutions for $\vec{x}$ arising from the method of characteristics, we were able to establish that $\vec{x} \in L^1$. Finally, this allowed us to establish a theorem which demonstrated that $\vec{x}$ is non-zero on the Cauchy horizon over some interval. This result pertains for all angular numbers $l$. 

Iguchi, Harada and Nakao \cite{INHeven} studied the behaviour of the quadrupole mode ($l=2$) of the even parity perturbations of the LTB spacetime. They numerically solved the linearised Einstein equations and found that this perturbation diverged on the Cauchy horizon. In one sense, our results are a generalization of theirs, in that this method allows us to treat all perturbations. However, the extra symmetry of self-similarity is needed in order to apply our methods. 

We note a potential issue with this work. Our perturbations are at linear order. It follows that it is somewhat strange to conclude that they diverge on the Cauchy horizon, as they are therefore far too large to remain at linear order. However, this result still indicates that this spacetime is not stable to perturbation. 

In previous work, we demonstrated that the odd parity perturbations of this background spacetime are finite for all $l$, where finiteness was measured with respect to initial values of a natural energy norm for the odd parity system. Taken as a whole, this work supports the hypothesis of cosmic censorship, in that one should expect perturbations on a naked singularity spacetime to diverge as the Cauchy horizon is approached. 

The background spacetime investigated here is of course not a serious model of gravitational collapse, as at the very least, it ignores the effects of pressure during the collapse. A natural extension of this work would therefore be to consider the self-similar perfect fluid model, which contains a naked singularity for a wide range of the equation of state parameter. The study of the behaviour of perturbations in this spacetime would be a very interesting application of the methods developed here. 

\ack
We thank Carsten Gundlach for helpful discussions. This research was funded by the Irish Research Council for Science, Engineering and Technology, grant number P02955. 

\appendix

\section{Equations of Motion}
\label{eqnsmotion}
We list here the various matrix coefficients and source terms omitted in sections \ref{sec:reduction}, \ref{sec:lqblowup} and \ref{sec:chdivbehaviour}. We neglect to list those terms whose exact form is not important. We note in what follows that $S=S(z)=(1+az)^{2/3}$ is the radial function, $q=q(z)=\lambda/4 \pi (-z \dot{S} +S)S^2$ is the density function and $\cdot=\frac{\partial }{\partial z}$. 

The five dimensional system takes the form 
\begin{equation*}
\frac{\partial \vec{u}}{\partial z} + \tilde{A}(z) \frac{\partial \vec{u}}{\partial p} + \tilde{C}(z) \vec{u} = \vec{\sigma}(z,p)
\end{equation*}
where the matrix coefficient $\tilde{A}(z)$ is given in Section \ref{sec:reduction}. $\tilde{C}(z)$ is given by  
\begin{equation*}
\tilde{C}=\left(
\begin{array}{ccccc}
c_{11} & c_{12} & c_{13} & c_{14} & c_{15} \\
c_{21} & c_{22} & 0 & c_{24} & c_{25} \\
c_{31} & c_{32} & c_{33} & c_{34} & c_{35}\\
c_{41} & 0 & c_{43} & c_{44} & c_{45} \\
c_{51} & 0 & c_{53} & c_{54} & c_{55}
\end{array}
\right), 
\end{equation*}
where 
\begin{eqnarray*}
\fl c_{11}=\frac{-3 z (-1+\dot{S}) \dot{S}^2-S^2 \ddot{S}+S (3 \dot{S}^2-z \ddot{S}+\dot{S} (-3+2 z \ddot{S}))}{S(-1+\dot{S}) (S-z \dot{S})}, \\ 
\fl c_{12}=-\frac{-1+\dot{S}}{2 S}, \qquad \qquad  \qquad \qquad  c_{13}= \frac{-1+\dot{S}}{2 S},  \\ 
\fl c_{14}= \frac{-1+\dot{S}}{2 S}, \qquad \qquad  \qquad \qquad  c_{15}= \frac{-1+\dot{S}}{2 S},  \\ 
\fl c_{21}= -\frac{8 \pi  q S}{-1+\dot{S}}, \qquad \qquad  \qquad \qquad  c_{22}= -\frac{\dot{q}}{q}, \\ 
\fl c_{24}= -\frac{4 \pi  q S}{1+\dot{S}}, \qquad \qquad  \qquad \qquad c_{25}= -\frac{4 \pi  q S}{-1+\dot{S}},  
\end{eqnarray*}

\begin{eqnarray*}
\fl c_{31}=\frac{n_{1}}{z S (-1+\dot{S}) (S-z \dot{S})^2}, \\ 
\fl n_{1}=2 (2 z^3 \dot{S}^4-2 z^2 S \dot{S}^2 (3 \dot{S}+z \ddot{S})+S^3 (-2 \dot{S}+z^2 S^{(3)}) \\ 
+z S^2 (6 \dot{S}^2+z^2 \ddot{S}^2+z \dot{S}(2 \ddot{S}-z S^{(3)}))), \\ 
\fl c_{32}=-\frac{2}{z}-\frac{\dot{q}}{q}, \qquad \qquad  \qquad \qquad c_{33}=\frac{1}{z}+\frac{2 \dot{S}}{S}, \\ 
\fl c_{34}=\frac{n_{2}}{2 z S (1+\dot{S})}, \\ 
\fl n_{2}=z \dot{S} (2-L-L^2+2 \dot{S}+2 \dot{S}^2)-4 S^2 \ddot{S}+S (-2+L+L^2-2 \dot{S}^2-2 z \ddot{S} \\ 
+\dot{S}(-2+4 z \ddot{S})), \\ 
\fl c_{35}=\frac{n_{3}}{2 z S (-1+\dot{S})}, \\ 
\fl n_{3}=z \dot{S} (-2+L+L^2+2 \dot{S}-2 \dot{S}^2)+4 S^2 \ddot{S}-S (-2+L+L^2-2 \dot{S}^2+2 z \ddot{S} \\ 
+\dot{S} (2+4 z \ddot{S})), 
\end{eqnarray*}

\begin{eqnarray*}
\fl c_{41}= \frac{4 (1+\dot{S}) (-z \dot{S}^2+S (\dot{S}+z \ddot{S}))}{S (-1+\dot{S}) (z-S+z \dot{S})}, \\ 
\fl c_{43}= \frac{(1+\dot{S}) (S-z \dot{S})}{S (z-S+z \dot{S})}, \\ 
\fl c_{44}= \frac{3+(z+S) \ddot{S}+\dot{S}(3+z \ddot{S})}{(1+\dot{S}) (z-S+z \dot{S})}, \\ 
\fl c_{45}= \frac{(1+\dot{S}) (S-z \dot{S})}{S (z-S+z \dot{S})}, 
\end{eqnarray*}

\begin{eqnarray*}
\fl c_{51}= -\frac{4 (-z \dot{S}^2+S (\dot{S}+z \ddot{S}))}{S (z+S-z \dot{S})}, \\ 
\fl c_{53}=-\frac{(-1+\dot{S}) (S-z \dot{S})}{S (z+S-z \dot{S})}, \\ 
\fl c_{54}=-\frac{(-1+\dot{S}) (S-z \dot{S})}{S (z+S-z \dot{S})}, \\ 
\fl c_{55}=-\frac{3+(-z+S) \ddot{S}+\dot{S} (-3+z \ddot{S})}{(-1+\dot{S}) (z+S-z \dot{S})}. 
\end{eqnarray*}
The five dimensional source term is given by $\vec{\Sigma}_{5}= \vec{f}(t)g(p)$ where $\vec{f}(t)=(0,0,f_{1}(z), f_{2}(z), f_{3}(z))^T$ and
\begin{eqnarray*}
f_{1}(z)=-\frac{16 \pi q (-S+z \dot{S})}{z S}, \qquad \qquad f_{2}(z)=-\frac{8 \pi  q(1+\dot{S}) (S-z \dot{S})}{S(z-S+z \dot{S})}, \\  f_{3}(z)=\frac{8 \pi q(-1+\dot{S})(S-z \dot{S})}{S[z](z+S-z \dot{S})}. 
\end{eqnarray*}
In sections \ref{sec:lqblowup} and \ref{sec:chdivbehaviour} we used the system in the form 
\begin{equation*}
t\frac{\partial \vec{u}}{\partial t} + A(t) \frac{\partial \vec{u}}{\partial p} + C(t) \vec{u} = \vec{\Sigma}(t,p). 
\end{equation*}
Here $A(t) = t \tilde{A}(z)$, $C(t)=t \tilde{C}(z)$ and $\vec{\Sigma}(t, p) = t\vec{\Sigma}_{5}(z,p)$. 

The coefficients appearing in (\ref{mnontriv}) are given by 
\begin{eqnarray*}
\fl g_{1}(z)=\frac{4 S(-6 z^2 S \dot{S}^3+2 z^3 \dot{S}^4+4 \pi  z q S^2 (S-z \dot{S})^2+S^3 (-2 \dot{S}+z \ddot{S})+z S^2 (6 \dot{S}^2-z \dot{S} \ddot{S}+z^2 \ddot{S}^2))}{(-1+\dot{S}) (-S+z \dot{S})}, \\ 
\fl g_{2}(z)=2 S(2 S^2+2 z^2 \dot{S}^2+z S (-4 \dot{S}+z \ddot{S})), \\ 
\fl g_{3}(z)=-2 S^2 (S+z (-\dot{S}+z \ddot{S})), \\ 
\fl g_{4}(z)=-\frac{m_{1}(z) }{1+\dot{S}}, \qquad  \qquad \qquad \qquad  g_{5}(z)=\frac{m_{2}(z)}{1-\dot{S}}, \\ 
\fl g_{6}(z)=2 S^2 (S-z \dot{S}), \qquad \qquad \qquad     g_{7}(z)=32 \pi  q S (S-z \dot{S})^2, 
\end{eqnarray*}
where 
\begin{eqnarray*}
\fl m_{1}(z)=S (z^2 \dot{S} (1+\dot{S}) (-2+L+L^2-2 \dot{S}-2 \dot{S}^2)+S^3 (8 \pi  z q-4 \ddot{S}) \\ 
-z S (-2+L+L^2-4 \dot{S}^3-2 z \ddot{S} +2 \dot{S} (-3+L+L^2+z \ddot{S})+\dot{S}^2 (-6+4 z \ddot{S})) \\ 
+S^2 (-2+L+L^2-2 \dot{S}^2+4 z \ddot{S}+\dot{S} (-2-8 \pi  z^2 q+8 z \ddot{S}))), \\ 
\fl m_{2}(z)=S (-z^2 (-1+\dot{S}) \dot{S} (-2+L+L^2+2 \dot{S}-2 \dot{S}^2)+4 S^3 (2 \pi  z q+\ddot{S}) \\ 
-z S (-2+L+L^2+4 \dot{S}^3+2 z \ddot{S} -2 \dot{S} (-3+L+L^2-z \ddot{S})-2 \dot{S}^2 (3+2 z \ddot{S})) \\ 
-S^2 (-2+L+L^2-2 \dot{S}^2-4 z \ddot{S}+\dot{S} (2+8 \pi  z^2 q+8 z \ddot{S}))).  
\end{eqnarray*}
The coefficient matrix $E(z)$ appearing in (\ref{4deqnofmotion}) is given by
\begin{equation*}
E(z)=\left( \begin{array}{cccc}
 0 & 1 & 0 & 0 \\
 0 & 0 & 0 & 0 \\
 0 & 0 & \frac{3 (1+a z)^{1/3}}{3+a z-3 z (1+a z)^{1/3}} & 0 \\
 0 & 0 & 0 & -\frac{3 (1+a z)^{1/3}}{3+a z+3 z (1+a z)^{1/3}}
\end{array}
\right), 
\end{equation*}
so we can explicitly see that this system is not symmetric hyperbolic. 

Finally, we present the similarity matrix $S$ (which appears in Section \ref{xbehaviour} and is used in the proof of Theorem \ref{Thm4}), a constant matrix which transforms the zero order term of the Taylor expansion of $C(t)$, $C(t=0)$, into $C_{0}$. It is given by
\begin{equation*}
S=\left(
\begin{array}{ccccc}
 s_{1} & s_{2}& s_{2} & 0 & 0 \\
 0 & 0 & 0 & 1 & 0 \\
 0 & 0 & 1 & 0 & 0 \\
 0 & 1 & 0 & 0 & 0 \\
 1 & 0 & 0 & 0 & 1
\end{array}
\right),
\end{equation*}
where
\begin{eqnarray*}
s_{1}=\frac{S(z_{c}) (3+(-z_{c}+S(z_{c})) \ddot{S}(z_{c})+\dot{S}(z_{c}) (-3+z_{c} \ddot{S}(z_{c})))}{4 (-1+\dot{S}(z_{c})) (z_{c} \dot{S}(z_{c})^2-S(z_{c}) (\dot{S}(z_{c})+z \ddot{S}(z_{c})))},  \\ 
s_{2}=\frac{(-1+\dot{S}(z_{c})) (-S(z_{c})+z_{c} \dot{S}(z_{c}))}{4 (-z_{c} \dot{S}(z_{c})^2+S(z_{c}) (\dot{S}(z_{c})+z_{c} \ddot{S}(z_{c})))}. 
\end{eqnarray*}

\section{The Behaviour of $c$}

In the statement of Theorem \ref{Thm2}, we looked that the Jordan canonical form of the zero order term in the Taylor expansion of the matrix $C(t)$. We claimed that the only non-zero eigenvalue of this matrix was in the range $c \in (3, +\infty)$, for $a \in (0, a^*)$. Here we prove this claim. 

In the previous section, we presented the coefficients of the equation of motion for $\vec{u}$. Recall that the coefficient of $\vec{u}$ was given as $\tilde{C}(z)$, where in terms of the notation of (\ref{5dsys}), $C(t) = t\tilde{C}(z)$. After we put $\tilde{C}(z)$ in Jordan canonical form, we find that the only non-zero term is the $(5, 5)$ entry, which is given by  
\begin{equation}
\label{cdef}
\tilde{c}_{55}=\frac{3+(S-z)\ddot{S}+ \dot{S}(-3+z \ddot{S})}{(1-\dot{S})(z+S-z\dot{S})}. 
\end{equation}
Now define $h(z):= z+S-z \dot{S}$. On the Cauchy horizon, $h(z_{c})=0$ by definition. We can use this to Taylor expand the numerator of (\ref{cdef}) and to simplify the zero order term. We can also Taylor expand the factor of $h(z)$ which appears in the denominator so that $h(z) = \dot{h}(z_{c})(z-z_{c})+ O((z-z_{c})^2) = (1-z_{c}\ddot{S}(z_{c}))+O((z-z_{c})^2)$. Inserting these expansions into (\ref{cdef}) produces 
\begin{equation}
\label{cdef2}
\tilde{c}_{55}=\left(2 + \frac{1}{1-z_{c}\ddot{S}(z_{c})} \right) \frac{1}{z-z_{c}} + O(1),
\end{equation}
as $z \rightarrow z_{c}$. So we must determine the value of $\dot{h}(z_{c})=1-z_{c}\ddot{S}(z_{c})=1+\frac{4}{3}a^2z_{c}(a)(1+az_{c}(a))^{-4/3}$, where we use (\ref{Sdef}) to prove the last equality and we write $z_{c}=z_{c}(a)$ to emphasise that the location of the Cauchy horizon depends on the value of $a$, the nakedness parameter. 

We first analyse how $z_{c}$ depends on $a$, before using this information to determine the behaviour of $\dot{h}(z_{c})$. We define $\lambda(z)=(1+az)h(z)^3$, so that $\lambda(z)=z^3(1+az)+(1+\frac{a}{3}z)^3$. The Cauchy horizon corresponds to the first negative root of the quartic equation $\lambda(z)=0$. This root must lie in the interval $z \in (-\frac{1}{a}, 0)$, since $z=-\frac{1}{a}$ corresponds to the singularity (see (\ref{Sdef})). We know that a root exists for $a \in (0, a^*)$, where $a=a^*$ corresponds to a double root of $\lambda$, where $\lambda(z_{c}(a^*))=\dot{\lambda}(z_{c}(a^*))=0$. Let $z^*:=z_{c}(a^*)$. We can easily show (using $\lambda = 0$) that 
\begin{equation*}
(1+\frac{a}{3}z)\dot{\lambda} \bigg|_{\lambda=0} = z^2(3+4az+\frac{1}{3}a^2z^2), 
\end{equation*}
which implies that $a^*z^*$ satisfies the quadratic 
\begin{equation}
\label{astardef}
(a^*)^2(z^*)^2 + 12a^*z^* +9 = 0. 
\end{equation}
By solving this for $a^*z^*$ and picking the larger root, we have
\begin{equation*}
z^*=\frac{1}{a^*}(-6 + 3\sqrt{3}). 
\end{equation*}
Next, we consider the dependence of $z_{c}$ on $a$. If we differentiate the condition $\lambda=0$ with respect to $a$, multiply by $(1+az)$ and use $\lambda=0$ to simplify, we find that 
\begin{equation}
\label{zcderiv}
(a^2z_{c}^2 + 12az_{c} + 9) \frac{dz_{c}}{da} = 2 az_{c}^3. 
\end{equation}
We note that the coefficient of $dz_{c} / da$ vanishes at $a=a^*$ (see \ref{astardef}) and is positive in the interval $a \in (0, a^*)$ (we can see this by calculating its roots and noting that for the range of allowable $a$, we are always above the larger root). Now since $z_{c}<0$, this implies that $dz_{c} / da < 0$ for all $a \in (0, a^*)$ and 
\begin{equation*}
\lim_{a \rightarrow a^*} \frac{d z_{c}}{da} = -\infty. 
\end{equation*}
So to summarise, we know that $z_{c}(0)=-1$ (see the definition of $\lambda$), $z_{c}(a^*)=z^*$ and $z_{c}$ is monotonically decreasing from $-1$ down to $z^*$ as $a$ increases from $0$ to $a^*$.

We now determine the range of $\dot{h}(z_{c})$. Define $u(a)=\dot{h}(z_{c}(a))$. Then $u(0)=1$ (since  $u(a)=\dot{h}(z_{c})=1+\frac{4}{3}a^2z_{c}(a)(1+az_{c}(a))^{-4/3}$), and by definition $u(a^*)=0$ (recall that $\lambda(a)=(1+az)h(a)^3$). A straightforward calculation (using (\ref{zcderiv})) shows that 
\begin{equation*}
\frac{d u}{da}=\frac{4}{9}az_{c}(1+az_{c})^{-4/3}\frac{(3+2az_{c})}{3+4az_{c}+\frac{1}{3}a^2z_{c}^2}. 
\end{equation*}
Since $a>0$, $z_{c}<0$, $1+az_{c}>0$ and $3+4az_{c}+\frac{1}{3}a^2z_{c}^2>0$ on $a \in (0, a^*)$, it follows that $du / da < 0$ for $a \in (0, a^*)$. So $u(a)=\dot{h}(z_{c}(a)) \in [0, 1]$ and $u$ decreases monotonically from $u=1$ at $a=0$ to $u=0$ at $a=a^*$. It follows that $\frac{1}{u(a)} \in (1, \infty)$ for $a \in (0, a^*)$ with 
\begin{equation*}
\lim_{a \rightarrow 0^{+}} \frac{1}{u(a)} = 1, \qquad \qquad \qquad \lim_{a \rightarrow (a^*)^{-}} \frac{1}{u(a)}= + \infty. 
\end{equation*}
We note that (\ref{cdef2}) can be written as 
\begin{equation*}
\tilde{c}_{55}=\left( 2+\frac{1}{u(a)} \right)\frac{1}{z-z_{c}}+O(1), 
\end{equation*}
and if we note that the $(5,5)$ entry in the Jordan canonical form of the zero order term in the Taylor expansion of $C(t)$ is $c:=c_{55} = t\tilde{c}_{55}$, where $t=z-z_{c}$, then we can conclude that $c=(2+\frac{1}{u(a)}) \in (3, +\infty)$ with 
\begin{equation*}
\lim_{a \rightarrow 0^{+}} c = 3, \qquad \qquad \qquad \lim_{a \rightarrow (a^*)^{-}} c = + \infty. 
\end{equation*}


\end{document}